\documentclass[11pt]{article}

\usepackage{amsmath}
\usepackage{amssymb, amsthm, amsfonts,bm,mathtools}
\usepackage[inline,shortlabels]{enumitem}
\usepackage{csquotes}
\usepackage{hhline}
\usepackage{cite}
\usepackage[numbers]{natbib}
\usepackage{framed}
\usepackage[framemethod=tikz]{mdframed}
\usepackage{graphicx}
\usepackage{color}

\usepackage{multirow}
\usepackage{algorithm, subcaption}
\PassOptionsToPackage{boxed,section}{algorithm}
\usepackage[noend]{algpseudocode}
\usepackage{amssymb}
\usepackage{float}
\setitemize{noitemsep,topsep=3pt,parsep=3pt,partopsep=3pt}

\usepackage[margin = 1in]{geometry}
\usepackage{adjustbox}
\usepackage{xspace}
\usepackage{thmtools}
\usepackage{thm-restate}

\definecolor{darkgreen}{rgb}{0,0.5,0}
\definecolor{darkblue}{rgb}{0,0,0.8}
\usepackage{hyperref}
\hypersetup{
    unicode=false,          % non-Latin characters in Acrobat’s bookmarks
    colorlinks=true,        % false: boxed links; true: colored links
    linkcolor=darkblue,          % color of internal links (change box color with linkbordercolor)
    citecolor=darkgreen,        % color of links to bibliography
    filecolor=magenta,      % color of file links
    urlcolor= black           % color of external links
}
\RequirePackage[]{silence}
\usepackage[nameinlink,capitalize]{cleveref}

\newtheorem{theorem}{Theorem}[section]

\newtheorem{lemma}{Lemma}[section]

\newtheorem{corollary}{Corollary}[section]

\newtheorem{observation}{Observation}[section]
\newtheorem{property}{Property}[section]

\theoremstyle{definition}
\newtheorem{definition}{Definition}[section]

\newcommand{\ignore}[1]{}

\algnewcommand\algorithmicswitch{\textbf{switch}}
\algnewcommand\algorithmiccase{\textbf{case}}

\algdef{SE}[SWITCH]{Switch}{EndSwitch}[1]{\algorithmicswitch\ #1\ \algorithmicdo}{\algorithmicend\ \algorithmicswitch}%
\algdef{SE}[CASE]{Case}{EndCase}[1]{\algorithmiccase\ #1}{\algorithmicend\ \algorithmiccase}%
\algtext*{EndSwitch}%
\algtext*{EndCase}%
\newcommand{\MCM}{{\sc mcm}}

\newcommand{\MWM}{{\sc mwm}}

\newcommand{\limit}{\operatorname{\text{\sc limit}}}
\newcommand{\MIS}{\operatorname{MIS}}

\newcommand{\congest}{\ensuremath{\mathsf{CONGEST}}\xspace}
\newcommand{\crewpram}{\ensuremath{\mathsf{CREW\ PRAM}}\xspace}
\newcommand{\local}{$\mathsf{LOCAL}$\xspace}
\newcommand{\suportedcongest}{\ensuremath{\mathsf{SUPPORTED ~CONGEST}~}}

\newcommand{\poly}{\operatorname{\text{{\rm poly}}}}

\newcommand{\polylog}{\operatorname{\text{{\rm polylog}}}}

\newcommand{\floor}[1]{\lfloor #1 \rfloor}

\newcommand{\cev}[1]{\reflectbox{\ensuremath{\vec{\reflectbox{\ensuremath{#1}}}}}}

\renewcommand{\paragraph}[1]{\vspace{0.15cm}\noindent {\bf #1}:}

\algnewcommand{\LineComment}[1]{\Statex \(\triangleright\) #1}
\newcommand{\Lmax}{\ell_{\max}}
\newcommand{\Cmax}{C_{\max{}}}
\newcommand{\Bmax}{C_{\max{}}}
\newcommand{\diamB}{\textrm{diam}(B)}
\newcommand{\FRemain}{\ensuremath{\tilde{F}}}
\newcommand{\GContract}{\ensuremath{\tilde{G}}}
\newcommand{\MContract}{\ensuremath{\tilde{M}}}
\newcommand{\AlgPhase}{{\textsc{Alg-Phase}}\xspace}
\newcommand{\VertexAlgPhase}{{\textsc{Vertex-Weighted-Alg-Phase}}\xspace}

\begin{document}
\title{$(1-\epsilon)$-Approximate Maximum Weighted Matching in\\ Distributed, Parallel, and Semi-Streaming Settings\footnote{A preliminary version appears in PODC 2023 as ``$(1-\epsilon)$-Approximate Maximum Weighted Matching in $\poly(1/\epsilon, \log n)$ Time in the Distributed and Parallel Settings''.}} 
\date{}

\setcounter{page}{0}
\author{
Shang-En Huang\thanks{Supported by NSF Grant No.~CCF-2008422.}\\
\normalsize Boston College \\
\normalsize \texttt{huangaul@bc.edu} \\
\and
Hsin-Hao Su\footnotemark[2]\\
\normalsize Boston College\\
\normalsize \texttt{suhx@bc.edu} \\
}

\maketitle
\begin{abstract}
 The maximum weighted matching (\MWM{}) problem is one of the most well-studied combinatorial optimization problems in distributed graph algorithms. Despite a long development on the problem, and the recent progress of Fischer, Mitrovic, and Uitto \cite{FMU22} who gave a $\poly(1/\epsilon, \log n)$-round algorithm for obtaining a $(1-\epsilon)$-approximate solution for unweighted maximum matching, it had been an open problem whether a $(1-\epsilon)$-approximate \MWM{} can be obtained in $\poly(1/\epsilon, \log n)$ rounds in the \congest model. Algorithms with such running times were only known for special graph classes such as bipartite graphs \cite{AKO18} and minor-free graphs \cite{CS22}. For general graphs, the previously known algorithms require exponential in $(1/\epsilon)$ rounds for obtaining a $(1-\epsilon)$-approximate solution \cite{FFK21} or achieve an approximation factor of at most 2/3 \cite{AKO18}. In this work, we settle this open problem by giving a deterministic $\poly(1/\epsilon, \log n)$-round algorithm for computing a $(1-\epsilon)$-approximate \MWM{} for general graphs in the \congest model.
 Our proposed solution extends the algorithm of Fischer, Mitrovic, and Uitto~\cite{FMU22}, blends in the sequential algorithm from Duan and Pettie~\cite{DuanP14} and the work of Faour, Fuchs, and Kuhn~\cite{FFK21}.
 Interestingly, this solution also implies a  CREW PRAM algorithm with $\poly(1/\epsilon, \log n)$ span using only $O(m)$ processors.
 In addition, 
 with the reduction from Gupta and Peng~\cite{GuptaP13}, we further obtain a semi-streaming algorithm with $\poly(1/\epsilon)$ passes.
 When $\epsilon$ is smaller than a constant $o(1)$ but at least $1/\log^{o(1)} n$, our algorithm 
 is more efficient than both Ahn and Guha's $\poly(1/\epsilon, \log n)$-passes algorithm~\cite{AhnG13} and Gamlath, Kale, Mitrovic, and Svensson's $(1/\epsilon)^{O(1/\epsilon^2)}$-passes algorithm~\cite{GKMS19}.

\end{abstract}
\thispagestyle{empty}	
\clearpage

\section{Introduction}

Matching problems are central problems in the study of both sequential and distributed graph algorithms.
A {\it matching} is a set of edges that do not share endpoints. Given a weighted graph $G = (V, E, w)$, where $w: E \to \{1, \ldots, W\}$, the maximum weight matching (\MWM{}) problem is to compute a matching $M$ with the maximum weight, where the weight of $M$ is defined as $\sum_{e \in M} w(e)$. Given an unweighted graph $G = (V, E)$, the maximum cardinality matching (\MCM{}) problem is to compute a matching $M$ such that $|M|$ is maximized. Clearly, the \MCM{} problem is a special case of the \MWM{} problem.  For $0 < \epsilon < 1$, a $(1-\epsilon)$-\MWM{} (or $(1-\epsilon)$-\MCM{}) is a $(1-\epsilon)$-approximate solution to the \MWM{} (or \MCM{}) problem. Throughout the paper, we let $n = |V|$ and $m = |E|$.

In distributed computing, the \MCM{} and \MWM{} problems have been studied extensively in the \congest model and the \local model. In these models, nodes host processors and operate in synchronized rounds. In each round, each node sends a message to its neighbors, receives messages from its neighbors, and performs local computations. The time complexity of an algorithm is defined to be the number of rounds used.  In the \local model, there are no limits on the message size, while the \congest model is a more realistic model where the message size is limited by $O(\log n)$ bits per link per round.

Computing an exact \MWM{} requires $\Omega(n)$ rounds in both the \congest model and the \local model (e.g., consider the graph $G$ to be a unit-weight even cycle.) Thus, the focus has been on developing efficient approximate algorithms.
In fact, the approximate \MWM{} problem is also one of the few classic combinatorial optimization problems where it is possible to bypass the notorious \congest model lower bound of $\tilde{\Omega}(D+\sqrt{n})$ by \cite{SarmaHKKNPPW12}, where $D$ denotes the diameter of the graph.
For $(1-\epsilon)$-\MWM{} in the \congest model, the lower bounds of \cite{KuhnMW16, AKO18} imply that polynomial dependencies on $(\log n)$ and $(1/\epsilon)$ are needed.
Whether matching upper bounds can be achieved is an intriguing and important  problem, as also mentioned in \cite{FFK21}:

\begin{quotation}``Obtaining a $(1-\epsilon)$-approximation (for \MWM{}) in $\poly(\log n/\epsilon)$ \congest rounds is one of the key open questions in understanding the distributed complexity of maximum matching.''
\end{quotation}

A long line of studies has been pushing progress toward the goal.
Below, we summarize the current fronts made by the existing results (also see \Cref{table:matching}).

\renewcommand{\arraystretch}{1.15}
\begin{table}[h!]\centering
\caption{Previous results on \MCM{} and \MWM{} in the \congest model and the \local model.  }\label{table:matching}
\begin{adjustbox}{width=\textwidth,center}
\begin{tabular}{llllll}
Citation                                          & Problem                                                                           & Ratio                            & Running Time                                                                                             & Type              & Model                        \\ \hline
\multicolumn{1}{|l|}{\cite{czygrinow2008fast}}                      & \multicolumn{1}{l|}{\begin{tabular}[c]{@{}l@{}}\MCM \\  (planar)\end{tabular}}                                                          & \multicolumn{1}{l|}{$1-\epsilon$}     & \multicolumn{1}{l|}{$O(\log(1/\epsilon)\cdot \log^{*} n)$}                                                                                    & \multicolumn{1}{l|}{Det.}  & \multicolumn{1}{l|}{\local}   \\ \hline
\multicolumn{1}{|l|}{\cite{CzygrinowHS09}}                      & \multicolumn{1}{l|}{\begin{tabular}[c]{@{}l@{}}\MCM \\  (bounded arb.)\end{tabular}}                                                          & \multicolumn{1}{l|}{$1-\epsilon$}     & \multicolumn{1}{l|}{$(1/\epsilon)^{O(1/\epsilon)}+O(\log^{*} n)$}                                                                                    & \multicolumn{1}{l|}{Det.}  & \multicolumn{1}{l|}{\local}   \\ \hline
\multicolumn{1}{|l|}{\cite{Nieberg08}}                     & \multicolumn{1}{l|}{\MWM}                                                          & \multicolumn{1}{l|}{$1-\epsilon$}     & \multicolumn{1}{l|}{$O(\epsilon^{-2} \log n \cdot T_{\MIS}(n^{O(1/\epsilon)}))$}                          & \multicolumn{1}{l|}{}      & \multicolumn{1}{l|}{\local}   \\ \hline
\multicolumn{1}{|l|}{\cite{BEPS16}}            & \multicolumn{1}{l|}{\MCM}                                                          & \multicolumn{1}{l|}{$\frac{1}{2}$}              & \multicolumn{1}{l|}{$O(\log \Delta + \log^{4} \log n)$}                                                  & \multicolumn{1}{l|}{Rand.} & \multicolumn{1}{l|}{\local}   \\ \hline
\multicolumn{1}{|l|}{\multirow{2}{*}{\begin{tabular}[c]{@{}l@{}} \cite{EvenMR15}\end{tabular}}}                         & \multicolumn{1}{l|}{\MCM}                                                          & \multicolumn{1}{l|}{$1-\epsilon$}              & \multicolumn{1}{l|}{$\Delta^{O(1/\epsilon)}+O(\log^{*} n / \epsilon^2)$}                                                     & \multicolumn{1}{l|}{Det.}  & \multicolumn{1}{l|}{\local} \\ \cline{2-6}
\multicolumn{1}{|l|}{}                            & \multicolumn{1}{l|}{\MWM}                                                          & \multicolumn{1}{l|}{$1-\epsilon$}   & \multicolumn{1}{l|}{
$O(\epsilon^{-2}\log \epsilon^{-1}) \cdot \log^\ast n + \Delta^{O(1/\epsilon)} \cdot O(\log \Delta)$
%$\log^{O(1/\epsilon)}W \cdot(\Delta^{O(1/\epsilon)}+O(\log^{*} n ))$
}                              & \multicolumn{1}{l|}{Det.}  & \multicolumn{1}{l|}{\local} \\ \hline
\multicolumn{1}{|l|}{\cite{FGK17}}              & \multicolumn{1}{l|}{\MCM}                                                          & \multicolumn{1}{l|}{$1-\epsilon$}     & \multicolumn{1}{l|}{$O(\Delta^{1/\epsilon} + \poly(\frac{1}{\epsilon}) \log^{*} n)$}                                                                                    & \multicolumn{1}{l|}{Det.}  & \multicolumn{1}{l|}{\local}   \\ \hline
\multicolumn{1}{|l|}{\multirow{2}{*}{\begin{tabular}[c]{@{}l@{}}\cite{GhaffariKM17} $+$\\ ~\cite{RozhonG20,GhaffariGR21} \end{tabular}}}             & \multicolumn{1}{l|}{\MWM}                                                          & \multicolumn{1}{l|}{$1-\epsilon$}     & \multicolumn{1}{l|}{$O(\epsilon^{-1} \log^3 n)$}                                                                                    & \multicolumn{1}{l|}{Rand.}  & \multicolumn{1}{l|}{\local}  \\ \cline{2-6}
\multicolumn{1}{|l|}{}                            & \multicolumn{1}{l|}{\MWM}                                                          & \multicolumn{1}{l|}{$1 - \epsilon$} & \multicolumn{1}{l|}{$O(\epsilon^{-1} \log^7 n)$}           & \multicolumn{1}{l|}{Det.}  & \multicolumn{1}{l|}{\local}  \\ \hline
\multicolumn{1}{|l|}{\cite{GHK18}}             & \multicolumn{1}{l|}{\MCM}                                                          & \multicolumn{1}{l|}{$1-\epsilon$}     & \multicolumn{1}{l|}{$O(\epsilon^{-9} \log^{5} \Delta \log^2 n)$}                                                                                    & \multicolumn{1}{l|}{Det.}  & \multicolumn{1}{l|}{\local}   \\ \hline
\multicolumn{1}{|l|}{\begin{tabular}[c]{@{}l@{}}\cite{GHK18} $+$\\ \cite{GhaffariKMU18} \end{tabular}}             & \multicolumn{1}{l|}{\MWM}                                                          & \multicolumn{1}{l|}{$1-\epsilon$}     & \multicolumn{1}{l|}{$O(\epsilon^{-7} \log^{4} \Delta \log^3 n)$}                                                                                    & \multicolumn{1}{l|}{Det.}  & \multicolumn{1}{l|}{\local}   \\ \hline
\multicolumn{1}{|l|}{\multirow{2}{*}{\begin{tabular}[c]{@{}l@{}} \\ \cite{Harris19}\end{tabular}}}                      & \multicolumn{1}{l|}{\MWM}                                                          & \multicolumn{1}{l|}{$1-\epsilon$}     & \multicolumn{1}{l|}{$O(\epsilon^{-4} \log^2 \Delta + \epsilon^{-1} \log^{*} n)$}                                                                                    & \multicolumn{1}{l|}{Det.}  & \multicolumn{1}{l|}{\local} \\ \cline{2-6}
\multicolumn{1}{|l|}{}                            & \multicolumn{1}{l|}{\MWM}                                                          & \multicolumn{1}{l|}{$1 - \epsilon$} & \multicolumn{1}{l|}{\begin{tabular}[c]{@{}l@{}}$O(\epsilon^{-3}\log (\Delta + \log \log n)$ \\$+\epsilon^{-2}\cdot (\log \log n)^2 )$\end{tabular}}           & \multicolumn{1}{l|}{Rand.}  & \multicolumn{1}{l|}{\local}  \\ \hline
\multicolumn{1}{|l|}{\cite{II86}}            & \multicolumn{1}{l|}{\MCM}                                                          & \multicolumn{1}{l|}{$\frac{1}{2}$}              & \multicolumn{1}{l|}{$O(\log n)$}                                                                         & \multicolumn{1}{l|}{Rand.} & \multicolumn{1}{l|}{\congest} \\ \hline
\multicolumn{1}{|l|}{\cite{ABI86} }                 & \multicolumn{1}{l|}{\MCM}                                                          & \multicolumn{1}{l|}{$\frac{1}{2}$}              & \multicolumn{1}{l|}{$O(\log n)$}                                                                         & \multicolumn{1}{l|}{Rand.} & \multicolumn{1}{l|}{\congest} \\ \hline
\multicolumn{1}{|l|}{\cite{Luby86}}                        & \multicolumn{1}{l|}{\MCM}                                                          & \multicolumn{1}{l|}{$\frac{1}{2}$}              & \multicolumn{1}{l|}{$O(\log n)$}                                                                         & \multicolumn{1}{l|}{Rand.} & \multicolumn{1}{l|}{\congest} \\ \hline
\multicolumn{1}{|l|}{\cite{HKP01}}           & \multicolumn{1}{l|}{\MCM}                                                          & \multicolumn{1}{l|}{$\frac{1}{2}$}              & \multicolumn{1}{l|}{$O(\log^4  n)$}                                                                         & \multicolumn{1}{l|}{Det.}  & \multicolumn{1}{l|}{\congest} \\ \hline
\multicolumn{1}{|l|}{\cite{WW04}} & \multicolumn{1}{l|}{\MWM}                                                          & \multicolumn{1}{l|}{$\frac{1}{5}$}              & \multicolumn{1}{l|}{$O(\log^2 n)$}                                                                       & \multicolumn{1}{l|}{Rand.} & \multicolumn{1}{l|}{\congest} \\ \hline
\multicolumn{1}{|l|}{\cite{LPR09}}               & \multicolumn{1}{l|}{\MWM}                                                          & \multicolumn{1}{l|}{$\frac{1}{4} - \epsilon$} & \multicolumn{1}{l|}{$O(\epsilon^{-1}\log \epsilon^{-1} \log n)$}                                         & \multicolumn{1}{l|}{Rand.}      & \multicolumn{1}{l|}{\congest} \\ \hline
\multicolumn{1}{|l|}{\multirow{3}{*}{\begin{tabular}[c]{@{}l@{}} \cite{LPP15}\end{tabular}}}                & \multicolumn{1}{l|}{\begin{tabular}[c]{@{}l@{}}\MCM\\  (bipartite)\end{tabular}} & \multicolumn{1}{l|}{$1-\epsilon$}     & \multicolumn{1}{l|}{$O(\log n/\epsilon^3)$}                                                              & \multicolumn{1}{l|}{Rand.} & \multicolumn{1}{l|}{\congest} \\ \cline{2-6}
\multicolumn{1}{|l|}{}                            & \multicolumn{1}{l|}{\MCM}                                                          & \multicolumn{1}{l|}{$1-\epsilon$}     & \multicolumn{1}{l|}{$2^{O(1/\epsilon)} \cdot O(   \epsilon^{-4} \log \epsilon^{-1} \cdot \log n)$} & \multicolumn{1}{l|}{Rand.}      & \multicolumn{1}{l|}{\congest} \\ \cline{2-6}
\multicolumn{1}{|l|}{}                            & \multicolumn{1}{l|}{\MWM}                                                          & \multicolumn{1}{l|}{$\frac{1}{2} - \epsilon$} & \multicolumn{1}{l|}{$O(\log(1/\epsilon) \cdot \log n)$}                                                  & \multicolumn{1}{l|}{Rand.} & \multicolumn{1}{l|}{\congest} \\ \hline
\multicolumn{1}{|l|}{\multirow{4}{*}{\begin{tabular}[c]{@{}l@{}} \cite{BCGS17} \end{tabular}}}          & \multicolumn{1}{l|}{\MWM}                                                          & \multicolumn{1}{l|}{$\frac{1}{2}$}              & \multicolumn{1}{l|}{$O(T_{\MIS}(n) \cdot \log W)$}                                                            & \multicolumn{1}{l|}{Rand.} & \multicolumn{1}{l|}{\congest} \\ \cline{2-6}
\multicolumn{1}{|l|}{}                            & \multicolumn{1}{l|}{\MWM}                                                          & \multicolumn{1}{l|}{$\frac{1}{2}$}              & \multicolumn{1}{l|}{$O(\Delta + \log n)$}                                                                & \multicolumn{1}{l|}{Det.}  & \multicolumn{1}{l|}{\congest} \\ \cline{2-6}
\multicolumn{1}{|l|}{}                            & \multicolumn{1}{l|}{\MWM}                                                          & \multicolumn{1}{l|}{$\frac{1}{2}-\epsilon$}   & \multicolumn{1}{l|}{$O(\log \Delta/ \log \log \Delta)$}                                                  & \multicolumn{1}{l|}{Rand.} & \multicolumn{1}{l|}{\congest} \\ \cline{2-6}
\multicolumn{1}{|l|}{}                            & \multicolumn{1}{l|}{\MCM}                                                          & \multicolumn{1}{l|}{$1-\epsilon$}     & \multicolumn{1}{l|}{$2^{O(1/\epsilon)}\cdot O(\log \Delta/ \log \log \Delta)$}                                                  & \multicolumn{1}{l|}{Rand.} & \multicolumn{1}{l|}{\congest} \\ \hline
\multicolumn{1}{|l|}{\multirow{2}{*}{\begin{tabular}[c]{@{}l@{}} \cite{Fischer17}\end{tabular}}}                         & \multicolumn{1}{l|}{\MCM}                                                          & \multicolumn{1}{l|}{$\frac{1}{2}$}              & \multicolumn{1}{l|}{$O(\log^2 \Delta \cdot \log n)$}                                                     & \multicolumn{1}{l|}{Det.}  & \multicolumn{1}{l|}{\congest} \\ \cline{2-6}
\multicolumn{1}{|l|}{}                            & \multicolumn{1}{l|}{\MWM}                                                          & \multicolumn{1}{l|}{$\frac{1}{2}-\epsilon$}   & \multicolumn{1}{l|}{$O(\log^2 \Delta \cdot \log \epsilon^{-1} + \log^{*} n)$}                              & \multicolumn{1}{l|}{Det.}  & \multicolumn{1}{l|}{\congest} \\ \hline
\multicolumn{1}{|l|}{\multirow{2}{*}{\begin{tabular}[c]{@{}l@{}} \cite{AKO18}\end{tabular}}}              & \multicolumn{1}{l|}{\begin{tabular}[c]{@{}l@{}}\MWM \\  (bipartite)\end{tabular}} & \multicolumn{1}{l|}{$1-\epsilon$}     & \multicolumn{1}{l|}{$O(\frac{\log(\Delta W )}{\epsilon^2} + \frac{\log^2 \Delta + \log^{*} n} {\epsilon})$}           & \multicolumn{1}{l|}{Det.}  & \multicolumn{1}{l|}{\congest} \\ \cline{2-6}
\multicolumn{1}{|l|}{}                            & \multicolumn{1}{l|}{\MWM}                                                          & \multicolumn{1}{l|}{$\frac{2}{3} - \epsilon$} & \multicolumn{1}{l|}{$O(\frac{\log(\Delta W )}{\epsilon^2} + \frac{\log^2 \Delta + \log^{*} n} {\epsilon})$}           & \multicolumn{1}{l|}{Det.}  & \multicolumn{1}{l|}{\congest} \\ \hline
\multicolumn{1}{|l|}{\cite{FFK21}}            & \multicolumn{1}{l|}{\MWM}                                                          & \multicolumn{1}{l|}{$1-\epsilon$}              & \multicolumn{1}{l|}{$2^{O(1/\epsilon)}\cdot \polylog(n)$}                                                                         & \multicolumn{1}{l|}{Det.} & \multicolumn{1}{l|}{\congest} \\ \hline
\multicolumn{1}{|l|}{\cite{FMU22}}            & \multicolumn{1}{l|}{\MCM}                                                          & \multicolumn{1}{l|}{$1-\epsilon$}              & \multicolumn{1}{l|}{$\poly(\log n, 1/\epsilon)$}                                                                         & \multicolumn{1}{l|}{Det.} & \multicolumn{1}{l|}{\congest} \\ \hline
\multicolumn{1}{|l|}{\cite{CS22}}                      & \multicolumn{1}{l|}{\begin{tabular}[c]{@{}l@{}}\MWM \\  (minor-free)\end{tabular}}                                                          & \multicolumn{1}{l|}{$1-\epsilon$}     & \multicolumn{1}{l|}{$\poly(\log n, 1/\epsilon)$}                                                                                    & \multicolumn{1}{l|}{Rand.}  & \multicolumn{1}{l|}{\congest}   \\ \hline
\multicolumn{1}{|l|}{{\bf this paper}}            & \multicolumn{1}{l|}{\MWM}                                                          & \multicolumn{1}{l|}{$1-\epsilon$}              & \multicolumn{1}{l|}{$\poly(\log n, 1/\epsilon)$}                                                                         & \multicolumn{1}{l|}{Det.} & \multicolumn{1}{l|}{\congest} \\ \hline
\end{tabular}
\end{adjustbox}
\end{table}
\renewcommand{\arraystretch}{1}

\begin{itemize}[leftmargin=*]
\item $c$-\MWM{} algorithms for $c < 2/3$. Wattenhofer and Wattenhofer \cite{WW04} were among the first to study the \MWM{} problem in the \congest model. They gave an algorithm for computing a $(1/5)$-\MWM{} that runs in $O(\log^2 n)$ rounds. Then Lotker, Patt-Shamir, and Ros\'en~\cite{LPR09} developed an algorithm that computes a $(1/4-\epsilon)$-\MWM{} in $O((1/\epsilon) \log (1/\epsilon) \log n)$ rounds. Later, Lotker, Patt-Shamir, and Pettie~\cite{LPP15} improved the approximation ratio and the number of rounds to $1/2-\epsilon$ and $O(\log(1/\epsilon)\cdot \log n)$ respectively.
Bar-Yehuda, Censor-Hillel, Gaffari, and Schwartzman~
\cite{BCGS17} gave a (1/2)-\MWM{} algorithm that runs in $O(T_{\MIS}(n) \cdot \log W)$ rounds, where $T_{\MIS}(n)$ is the time needed to compute a maximal independent set (MIS) in an $n$-node graph. Fischer \cite{Fischer17} gave a deterministic algorithm that computes a $(1/2 -\epsilon)$-\MWM{} in $O(\log^{2}\Delta \cdot \log \epsilon^{-1} + \log^{*} n)$ rounds by using a rounding approach, where $\Delta$ is the maximum degree. Then Ahmadi, Khun, and Oshman~\cite{AKO18} gave another rounding approach for $(2/3 - \epsilon)$-\MWM{} that runs in $O(\frac{\log(\Delta W)}{\epsilon^2} + \frac{\log^2 \Delta + \log^{*} n}{\epsilon})$ rounds deterministically. The rounding approaches of  \cite{Fischer17} and \cite{AKO18} inherently induce a $2/3$ approximation ratio because the linear programs they consider have an integrality gap of 2/3 in general graphs.

\item Exponential-in-$(1/\epsilon)$ algorithms. \cite{LPP15} showed that the random bipartition technique can be applied to get a randomized $2^{O(1/\epsilon)} \cdot O(\log n)$-round $(1-\epsilon)$-\MCM{} algorithm. 
Such a technique was later also applied by \cite{FFK21}, who gave a deterministic $2^{O(1/\epsilon)} \cdot \poly(\log n)$-round algorithm for $(1-\epsilon)$-\MWM{}. 

\item Bipartite graphs and other special graphs. For bipartite graphs, Lotker et al.~\cite{LPP15} gave an algorithm for $(1-\epsilon)$-\MCM{} that runs in $O(\log n /\epsilon^3)$ rounds. Ahmadi et al.~\cite{AKO18} showed that $(1-\epsilon)$-\MWM{} in bipartite graphs can be computed in $O(\log(\Delta W) / \epsilon^2 + (\log^2 \Delta + \log^{*} n) /\epsilon)$ rounds deterministically. Recently, Chang and Su~\cite{CS22} showed that a $(1-\epsilon)$-\MWM{} can be obtained in $\poly(1/\epsilon, \log n)$ rounds in minor-free graphs with randomization by using expander decompositions.  

\item Algorithms using larger messages. It was shown in~\cite{FGK17} that the $(1-\epsilon)$-\MWM{} problem can be reduced to hypergraph maximal matching problems, which are known to be solvable efficiently in the \local model. A number of $\poly(1/\epsilon, \log n)$-round algorithms are known for obtaining $(1-\epsilon)$-\MWM{} \cite{Nieberg08,  GhaffariKMU18, Harris19}. The current fastest algorithms are by \cite{Harris19}, who gave a $O(\epsilon^{-3} \log(\Delta + \log \log n)+\epsilon^{-2}(\log \log n)^2)$-round randomized algorithm and a $O(\epsilon^{-4} \log^2 \Delta + \epsilon^{-1} \log^{*} n)$-round deterministic algorithm.

\end{itemize}
Recently, Fischer, Mitrovi\'c, and Uitto~\cite{FMU22} made significant progress by giving a $\poly(1/\epsilon, \log n)$-round algorithm for computing a $(1-\epsilon)$-\MCM~{} --- the unweighted version of the problem.
Despite the progress,
the complexity of $(1-\epsilon)$-\MWM{} still remains unsettled. 
We close the gap by giving the first  $\poly(1/\epsilon, \log n)$ round algorithm for computing $(1-\epsilon)$-\MWM{} in the \congest model. 
The result is summarized as \cref{thm:main}.

\begin{theorem}\label{thm:main}
  There exists a deterministic \congest{} algorithm that solves the $(1-\epsilon)$-\MWM{} problem in $\poly(1/\epsilon, \log n)$ rounds.
\end{theorem}

In the parallel setting, Hougardy and Vinkemeier~\cite{HougardyV06} gave a \textsf{CREW PRAM}\footnote{A parallel random access machine that allows concurrent reads but requires exclusive writes.} algorithm that solves the $(1-\epsilon)$-\MWM{} problem in $O(\frac1{\epsilon}\log^5 n)$ span with $n^{O(1/\epsilon)}$ processors.  However, it is still not clear whether a {\it work-efficient} algorithm with a $\poly(1/\epsilon, \log n)$-span and $O(m)$ processors exists. Our \congest algorithm can be directly simulated in the \textsf{CREW PRAM} model, obtaining a $\poly(1/\epsilon, \log n)$ span algorithm that uses only $O(m)$ processors. The total work matches the best known sequential algorithm of \cite{DuanP14}, up to $\poly(1/\epsilon, \log n)$ factors.

\begin{corollary}\label{thm:main-parallel}
  There exists a deterministic \textsf{CREW PRAM} algorithm that solves the $(1-\epsilon)$-\MWM{} problem with $\poly(1/\epsilon, \log n)$ span and uses only $O(m)$ processors.
\end{corollary}

\paragraph{Semi-Streaming Model} 
In the semi-streaming model, the celebrated results of $(1-\epsilon)$-\MWM{} with $\poly(1/\epsilon, \log n)$ passes were already known by Ahn and Guha~\cite{AhnG13, AhnG11}.
Thus, in the semi-streaming model, the focus has been on obtaining algorithms with $o(\log n)$ dependencies on $n$.
The state of the art algorithms for $(1-\epsilon)$-\MWM{} still have exponential dependencies on $(1/\epsilon)$ (see \cite{GKMS19}).
Recently, Fischer, Mitrovi\'{c}, and Uitto~\cite{FMU22} made a breakthrough in the semi-streaming model, obtaining a $\poly(1/\epsilon)$  passes algorithm to the $(1-\epsilon)$-\MCM{} problem.

Our \congest algorithm translates to an $\poly(1/\epsilon)\cdot\log W$ passes algorithm in the semi-streaming model.
Bernstein and Dudeja~\cite{BD23} pointed out that,
with the reduction from Gupta and Peng~\cite{GuptaP13}, an input instance can be reduced into $O(\log_{1/\epsilon} W)$ instances of $(1-O(\epsilon))$-\MWM{} such that, the largest weight $W'$ in each instance can be upper bounded by $W'=(1/\epsilon)^{O(1/\epsilon)}$.
%Despite the additional $O(\log W)$ factor in the memory consumption,
Now that $\log W'=\poly(1/\epsilon)$, 
by running all the $(1-O(\epsilon))$-\MWM{} instances in parallel, we obtain
the very first $\poly(1/\epsilon)$-passes semi-streaming algorithm that computes an $(1-\epsilon)$-\MWM{}.
We summarize the result below, and provide the proof in appendix.

\begin{theorem}\label{thm:semi-streaming-result}
  There exists a deterministic algorithm that returns an $(1-\epsilon)$-approximate maximum weighted matching using $\poly(1/\epsilon)$ passes in the semi-streaming model. The algorithm requires $O(n\cdot \log W\cdot  \poly(1/\epsilon))$ words of memory.
\end{theorem}

%It is not known yet whether there is a $\poly(1/\epsilon)$ passes algorithm for $(1-\epsilon)$-\MWM{}. 

We remark that the results of Ahn and Guha~\cite{AhnG13, AhnG11} do not translate easily to a \congest{} algorithm within $\poly(1/\epsilon, \log n)$ rounds.
In particular, in \cite{AhnG11} the algorithm reduces to solving several instances of minimum odd edge cut\footnote{The goal is to return a mincut $(X, V\setminus X)$ among all subsets $X\subseteq V$ with an odd cardinality and  $|X|=O(1/\epsilon)$.}.
It seems hard to solve minimum odd edge cut in \congest, given the fact that approximate minimum edge cut has a lower bound $\tilde\Omega(D+\sqrt{n})$~\cite{GhaffariK13}, where $D$ is the diameter of the graph.
%Ahn and Guha built Gomory-Hu trees and
On the other hand, in \cite{AhnG13} the runtime per pass could be as high as $n^{O(1/\epsilon)}$, so it would be inefficient in \congest.

\subsection{Related Works and Other Approaches}
\label{sec:tech_sum}
\paragraph{Sequential Model} For the sequential model,  by the classical results of \cite{MV80,Blum90,GT91}, it was known that the exact \MCM{} and \MWM{} problems can be solved in $\tilde{O}(m\sqrt{n})$ time. For approximate matching, it is well-known that a $\frac12$-\MWM{} can be computed in linear time by computing a maximal matching.
Although near-linear time algorithms for $(1-\epsilon)$-\MCM{} were known in the 1980s~\cite{MV80,GT91},
it was a challenging task to obtain a near-linear time $\alpha$-\MWM{} algorithm for the approximate ratio $\alpha > \frac12$.
Several near-linear time algorithms were developed, such as $(\frac23-\epsilon)$-\MWM{}~\cite{DH03a,PS04} and $(\frac34-\epsilon)$-\MWM{}~\cite{DuanP10,HankeH10}.
Duan and Pettie~\cite{DuanP14} gave the first 
near-linear time algorithms for $(1-\epsilon)$-\MWM{}, which runs in $O(\epsilon^{-1} \log (1/\epsilon)\cdot m)$ time.

 \paragraph{Other Approaches} A number of different approaches have been proposed for the $(1-\epsilon)$-\MWM{} problem in distributed settings, which we summarize and discuss as follows: 
 
 \begin{itemize}[leftmargin = *]
 
 \item Augmenting paths. %For the unweighted \MCM{} problem, it is well known if there is no augmenting path of length $O(1/\epsilon)$ with respect to a matching, then it is a $(1-\epsilon)$-\MCM{}.  
 We say an augmenting path is an $l$-{\it augmenting path} if it contains at most $l$ vertices. Being able to find a set of (inclusion-wise) maximal  $l$ augmenting paths in $\poly(l, \log n)$ rounds is a key subproblem in many known algorithms for $(1-\epsilon)$-\MCM{},  where $l = O(1/\epsilon)$. In bipartite graphs,  \cite{LPP15} showed that the subproblem can be done by simulating Luby's MIS algorithm on the fly. On general graphs, finding an augmenting path is significantly more complicated than that in bipartite graphs. Finding a maximal set of augmenting paths is even more difficult. In the recent breakthrough of \cite{FMU22}, they showed how to find an ``almost'' maximal set of $l$-augmenting paths in $\poly(l, \log n)$ rounds in the \congest model in general graphs via bounded-length parallel DFS.  We note that the problem of finding a maximal set of $l$ augmenting paths can be thought of as finding a hypergraph maximal matching, where an $l$-augmenting path is represented by a rank-$l$ hyperedge.

 \item Hypergraph maximal matching. For the \MWM{} problem, the current approaches \cite{HougardyV06, Nieberg08, GhaffariKMU18, Harris19} in the \textsf{PRAM} model and the \local model  consider an extension of $l$-augmenting paths, the {\it $l$-augmentations}. Roughly speaking, an $l$-augmentation is an alternating path\footnote{more precisely, with an additional condition that each endpoint is free if its incident edge is unmatched.} or cycle with at most $l$ vertices. Similar to $l$-augmenting paths, the $l$-augmentations can also be represented by a rank-$l$ hypergraph (albeit a significantly larger one). Then they divide the augmentations into poly-logarithmic classes based on their {\it gains}. From the class with the highest gain to the lowest, compute the hypergraph maximal matching of the hyperedges representing those augmentations. While in the \local model and the \textsf{PRAM} model, the rank-$l$ hypergraph can be built explicitly and maximal independent set algorithms can be simulated on the hypergraph efficiently to find a maximal matching; it is not the case for the \congest model due to the bandwidth restriction.

 \item The rounding approach. The rounding approaches work by first solving a linear program relaxation of the \MWM{} problem. In \cite{Fischer17, AKO18}, they both developed procedures for obtaining fractional solutions and deterministic procedures to round a fractional matching to an integer matching (with some loss). While \cite{AKO18} obtained an algorithm for $(1-\epsilon)$-\MWM{} in bipartite graphs, the direct linear program that they have considered has an integrality gap of 2/3 in general graphs. Therefore, the approximation factor will be inherently stuck at 2/3 without considering other formulations such as Edmonds' blossom linear program \cite{Edmonds65}. 
 
 \item The random bipartition approach. Bipartite graphs are where the matching problems are more well-understood. The random bipartition approach randomly partitions vertices into two sets and then ignores the edges within the same partition. A path containing $l$ vertices will be preserved with probability at least $2^{-l}$. By using this property, \cite{LPP15} gave a $(1-\epsilon)$-\MCM{} algorithm that runs in $2^{O(1/\epsilon)} \cdot O(\log n)$ rounds and \cite{FFK21} gave a $(1-\epsilon)$-\MWM{} algorithm that runs in $2^{O(1/\epsilon)} \cdot \polylog(n)$ rounds. Note that this approach naturally introduces an exponential dependency on $(1/\epsilon)$. 

 \end{itemize}

\subsection{Technique Overview}
Our approach is to parallelize Duan and Pettie's \cite{DuanP14} near-linear time algorithm, which involves combining the recent approaches of \cite{CS22} and \cite{FMU22} as well as several new techniques.  The algorithm of \cite{DuanP14} is a primal-dual based algorithm that utilizes Edmonds' formulation \cite{Edmonds65}. Roughly speaking, the algorithm maintains a matching $M$, a set of active blossoms $\Omega \subseteq 2^{V}$, dual variables $y: V \to \mathbb{R}$ and  $z: 2^{V} \to \mathbb{R}$ (see \cref{sec:preliminaries} for details of blossoms). It consists of $O(\log W)$ scales with exponentially decreasing step sizes. Each scale consists of multiple primal-dual iterations that operate on a contracted {\bf unweighted} subgraph, $G_{elig}/ \Omega$, which they referred to as the {\it eligible graph}. For each iteration in scale $i$, it tries to make progress on both the primal variables ($M$, $\Omega$) and the dual variables ($y,z$) by the step size of the scale.

Initially, $\Omega = \emptyset$ so no blossoms are contracted. The first step in adjusting the primal variable is to search for an (inclusion-wise) maximal set of augmenting paths in the eligible graph and augment along them. After the augmentation, their edges will disappear from the eligible graph. Although \cite{DuanP14} showed that such a step can be performed in linear time in the sequential setting, it is unclear how it can be done efficiently in $\poly(1/\epsilon, \log n)$ time in the \congest model or the \textsf{PRAM} model. Specifically, for example, it is impossible to find the augmenting paths of length $\Theta(n)$ in \Cref{fig:1} in such time in the \congest model.

Our first ingredient is an idea from \cite{CS22}, where they introduced the weight modifier $\Delta w$ and dummy free vertices to effectively remove edges and free vertices from the eligible graph. They used this technique to integrate the expander decomposition procedure into the algorithm of \cite{DuanP14} for minor-free graphs. As long as the total number of edges and free vertices removed is small, one can show that the final error can be bounded.

With this tool introduced, it becomes more plausible that a maximal set of augmenting paths can be found in $\poly(1/\epsilon, \log n)$ time, as we may remove edges to cut the long ones.  Indeed, in {\it bipartite graphs}, this can be done by partitioning matched edges into layers. An edge is in the $i$-th layer if the shortest alternating path from any free vertex that ends at it contains exactly $i$ matched edges. Let $M_i$ be the set of matched edges of the $i$-th layer. It must be that the removal of $M_{i}$ disconnects all augmenting paths that contain more than $i$ matched edges. Let $i^{*} = \arg \min_{1 \leq i \leq 1/\epsilon} |M_i|$ and thus $|M_{i^{*}}| \leq \epsilon |M|$.  The removal of $M_{i^{*}}$ would cause all the leftover augmenting paths to have lengths of $O(1/\epsilon)$.

In general graphs, the above {\it path-cutting technique} no longer works. The removal of $M_i$ would not necessarily disconnect augmenting paths that contain more than $i$ matched edges.  
Consider the example in \Cref{fig:3}: 
for any matched edge $e$, the shortest alternating path from a free vertex that ends at $e$ contains at most $2$ matched edges.
There is a (unique) augmenting path from $\alpha$ to $\beta$ with $12$ matched edges.
However,
the removal of $M_5$ (notice that $5 < 12$) would not disconnect this augmenting path, since $M_5 = \emptyset$.
One of the technical challenges is to have an efficient procedure to {\it find a small fraction of edges whose removal cut all the remaining long augmenting paths in general graphs}. 

Secondly, the second step of the primal-dual iterations of \cite{DuanP14} is to find a maximal set of full blossoms reachable from free vertices and add them to $\Omega$ so they become contracted in the eligible graph. The problem here is that such a blossom can have a size as large as $\Theta(n)$ (See \Cref{fig:2}), so contracting it would take $\Theta(n)$ time in the \congest model. So the other technical challenge is {\it to ensure such blossoms will not be formed, possibly by removing a small fraction of edges and free vertices.} In general, these technical challenges are to remove a small fraction of edges and free vertices to achieve the so-called {\it primal blocking condition}, which we formally define in \Cref{prop:PBC}. 

Note that the challenge may become more involved after the first iteration, where $\Omega$ is not necessarily empty. It may be the case that a blossom found in $G_{elig}/\Omega$ contains a very small number of vertices in the contracted graph $G_{elig}/\Omega$ but is very large in the original graph $G$. In this case, we cannot add it to $\Omega$ either, as it would take too much time to simulate algorithms on $G_{elig}/\Omega$ in the \congest model if $\Omega$ has a blossom containing too many vertices in $G$. Therefore, we also need to ensure such a blossom is never formed.

\begin{figure}[h]
\centering
\begin{subfigure}[b]{0.32\textwidth}
\centering
\hspace*{0.5cm}
\includegraphics[width=4.5cm]{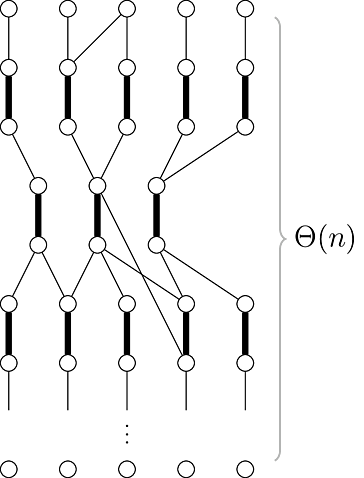}
\hspace*{-0.5cm}
\caption{}\label{fig:1}
\end{subfigure}
\begin{subfigure}[b]{0.32\textwidth}
\centering
%\vspace*{-0.5cm}
\includegraphics[width=3.8cm]{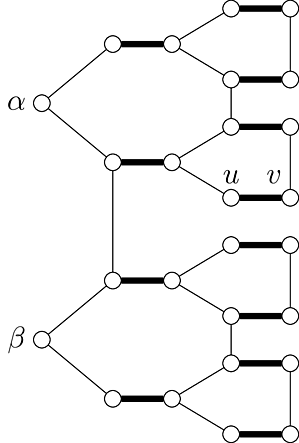}
%\vspace*{0.5cm}
\caption{}\label{fig:3}
\end{subfigure}
\begin{subfigure}[b]{0.32\textwidth}
\centering
\vspace*{-0.8cm}
\includegraphics[width=4.5cm]{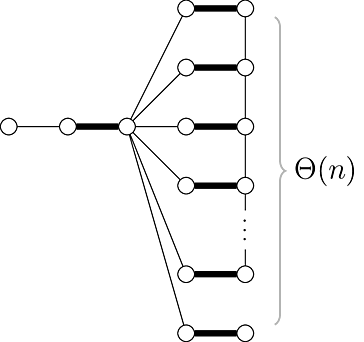}
\vspace*{0.8cm}
\caption{}\label{fig:2}
\end{subfigure}
\caption{Note that in these examples, we have $\Omega = \emptyset$ and so that $G/\Omega = G$.}
\end{figure}

To overcome these challenges, our second ingredient is the parallel DFS algorithm of Fischer, Mitrovic, and Uitto \cite{FMU22}. In \cite{FMU22}, they developed a procedure for finding an almost maximal set of $k$-augmenting paths in $\poly(k)$ rounds, where a $k$-augmenting path is an augmenting path of length at most $k$.  We show that the path-cutting technique  for bipartite graphs can be combined seamlessly with a tweaked, {\it vertex-weighted version} of \cite{FMU22} to overcome these challenges for general graphs.

The central idea of \cite{FMU22} is parallel DFS \cite{GPV93}. A rough description of the approach of \cite{FMU22} is the following: Start a bounded-depth DFS from each free vertex where the depth is bounded by $O(k)$ and each search maintains a cluster of vertices. The clusters are always vertex-disjoint. In each step, each search tries to enlarge the cluster by adding the next edge from its active path. If there is no such edge, the search will back up one edge on its active path. If the search finds an augmenting path that goes from one cluster to the other, then the two clusters are removed from the graph. Note that this is a very high-level description for the purpose of understanding our usage,  the actual algorithm of \cite{FMU22} is much more involved. For example, it could be possible that the search from one cluster overtakes some portion of another cluster.

The key property shown in \cite{FMU22} is that at any point of the search all the remaining $k$-augmenting paths must pass one of the edges on the active paths, so removing the edges on active paths of the searches (in addition to the removal of clusters where augmenting paths are found) would cut all $k$-augmenting paths. Moreover,  after searching for $\poly(k)$ steps, it is shown at most $1/\poly(k)$ fraction of searches remain active.
Since each DFS will only search up to a depth of $O(k)$, the number of edges on the active paths is at most $O(k)\cdot 1/\poly(k) = 1/\poly(k) $ fraction of the searches. In addition, we note that the process has an extra benefit that, roughly speaking, if a blossom is ought to be contracted in the second step of \cite{DuanP14}, it will lie entirely within a cluster or it will be far away from any free vertices.

To better illustrate how we use \cite{FMU22} to overcome these challenges, we first describe our procedure for the first iteration of \cite{DuanP14}, where $\Omega = \emptyset$. In this case, we run several iterations \cite{FMU22} to find a collection of $k$-augmenting paths, where $k = \poly(1/\epsilon, \log n)$, until the number of $k$-augmenting paths found is relatively small. Then remove (1) the clusters where augmenting paths have been found and (2) the active paths in the still active searches. By removing a structure, we meant using the weight modifier technique from \cite{CS22} to remove the matched edges and free vertices inside the structure. 

At this point, all the $k$-augmenting paths either overlap within the collection of $k$-augmenting paths or have been cut. The remaining augmenting paths must have lengths more than $k$.
To cut them, we contract all the blossoms found within each cluster. As the search only runs for $\poly(k)$ steps, each cluster has at most $\poly(k)$ vertices so these blossoms can be contracted in each cluster on a vertex locally by aggregating the topology to the vertex in $\poly(k)$ rounds.
The key property we show is that after the contraction, if we assign each blossom a weight proportional to its size, the weighted $O(k)$-neighborhood of the free vertices becomes bipartite.
The reason why this is correct is that the weighted distance is now an overestimate of the actual distance, and there are no full blossoms reachable within distance $k$ from the free vertices in the graph now.
Since the weighted $O(k)$-neighborhood from the free vertices are bipartite, we can run the aforementioned, but a weighted version, path-cutting technique on it to remove some edges augmenting paths of weighted length more than $k$. The weight assignment to the blossoms ensures that we will only remove a small fraction of the edges. 

Starting from the second iteration of \cite{DuanP14}, the set of active blossoms $\Omega$ may not be empty anymore.
We will need to be careful to not form any large nested blossoms after the Fischer-Mitrovic-Uitto parallel DFS algorithm (FMU-search), where the size of a blossom is measured by the number of vertices it contains in the original graph. To this end, when running the FMU-search, we run a weighted version of it, where each contracted vertex in $G_{elig}/\Omega$ is weighted proportional to the number of vertices it represents in the original graph. This way we can ensure the weight of each cluster is $\poly(k)$ and so the largest blossom it can form will be $\poly(k)$.

In order to generalize the properties guaranteed by FMU-search, one may have to open up the black-box and redo the whole sophisticated analysis of \cite{FMU22}.
However, we show that the properties can be guaranteed by a blossom-to-path simulation analysis, where each weighted blossom is replaced by an unweighted path. The properties guaranteed by FMU-search from the transformed unweighted graph can then be carried back to the blossom-weighted graph.

\paragraph{Organization} In \cref{sec:preliminaries}, we define the basic notations and give a brief overview of the scaling approach of \cite{DuanP14} as well as the modification of \cite{CS22}. In \cref{sec:scaling}, we describe our modified scaling framework. In \cref{sec:FMU}, we describe how \cite{FMU22} can be augmented to run in contracted graphs where vertices are weighted. In \cref{sec:MWMcongest}, we describe our \textsc{Approx\_Primal} procedure for achieving the primal blocking conditions.

\section{Preliminaries and Assumptions}\label{sec:preliminaries}

Throughout the paper, we denote $G=(V, E, \hat{w})$ to be the input weighted undirected graph, with an integer weight function $\hat{w}:E\to\{1, 2, \ldots, W\}$.

\paragraph{Matchings and Augmenting Paths} Given a matching $M$, a vertex is {\it free} if it is not incident to any edge in $M$. An {\it alternating path} is a path whose edges alternate between $M$ and $E \setminus M$.
An {\it augmenting path} $P$ is an alternating path that begins and ends with free vertices. Given an augmenting path $P$, let $M\oplus P = (M \setminus P) \cup (P \setminus M)$ denote the resulting matching after we augment along $P$. Note that we must have $|M \oplus P| = |M| + 1$.

\paragraph{Linear Program for \MWM{}}
Edmonds~\cite{Edmonds65} first formulated the matching polytope for general graphs.
On top of the bipartite graph linear programs, there are additional exponentially many constraints over $\mathcal{V}_{odd}$ --- all odd sized subsets of vertices. 
In this paper, we follow Edmonds'~\cite{Edmonds65} linear program formulation for  \MWM{} for the graph $(V, E, w)$:
\[
\begin{array}{c|c}
    \begin{array}{c}
    \textbf{Primal}\\[8pt]
    \begin{array}{rll}
        \text{max } &\multicolumn{2}{l}{ \sum_{e\in E} w(e)x(e) } \\[5pt]
        \text{st. } & \forall u\in V, & \sum_{v} x(uv)\le 1 \\[5pt]
        & \forall B\in \mathcal{V}_{odd}, & \sum_{u, v\in B} x(uv) \le  \frac{|B|-1}{2}\\[5pt]
        &\multicolumn{2}{l}{ x(e)\ge 0\ \forall e\in E }
    \end{array}
    \end{array}
    &
    \begin{array}{c}
    \textbf{Dual}\\[8pt]
    \begin{array}{rll}
        \text{min } &\multicolumn{2}{l}{ 
            \sum_{u\in V} y(u) + \sum_{B\in\mathcal{V}_{odd}} \frac{|B|-1}{2}z(B)
         } \\[5pt]
        \text{st. } & \forall uv\in E, & \\
        & \multicolumn{2}{l}{
        \ \ \ \ y(u)+y(v) + \sum_{B\ni u, v} z(B) \ge w(uv) } \\[10pt]
        & \multicolumn{2}{l}{ y(u) \ge 0, z(B) \ge 0 }
    \end{array} 
\end{array}
\end{array}
\]

\paragraph{Dual Variables} The variables $y(u)$ and $z(B)$ are called the {\it dual variables}.
For convenience, given an edge $e = uv$, we define $$yz(e)=y(u)+y(v)+\sum_{B \in \mathcal{V}_{odd} : e \in E(B)} z(B).$$

\paragraph{Blossoms} A blossom is specified with a vertex set $B$ and an edge set $E_{B}$.  A trivial blossom is when $B = \{v \}$ for some $v \in V$ and $E_{B}=\emptyset$. A non-trivial blossom is defined recursively: If there are an odd number of blossoms $B_0 \ldots B_{\ell}$ connected as an odd cycle by $e_i \in B_{i} \times B_{[(i+1)\mod (\ell+1)]}$ for $0 \leq i \leq \ell$, then $B = \bigcup_{i=0}^{\ell} B_i$ is a blossom with $E_B = \{e_0 \ldots, e_{\ell} \} \cup \bigcup_{i=0}^{\ell} E_{B_i} $. It can be shown inductively that $|B|$ is odd and so $B \in \mathcal{V}_{odd}$. A blossom is {\it full} if $|M \cap E_{B}| = (|B| -1 )/2$. The only vertex that is not adjacent to the matched edges in a full blossom is called the {\it base} of $B$. Note that $E(B) = \{(u,v) \mid u,v \in B\}$ may contain edges not in $E_{B}$. 

\paragraph{Active Blossoms}
A blossom is \emph{active} whenever $z(B) > 0$. We use $\Omega$ to denote the set of active blossoms throughout the execution of the algorithm.
Throughout the execution, we maintain the property that
only full blossoms will be contained in $\Omega$. Moreover, the set of active blossoms $\Omega$ forms a laminar (nested) family, which can be represented by a set of rooted trees. The leaves of the trees are the trivial blossoms. If a blossom $B$ is defined to be the cycle formed by $B_0,\ldots, B_{\ell}$, then $B$ is the parent of $B_0,\ldots, B_{\ell}$. The blossoms that are represented by the roots are called the {\it root blossoms}.

\paragraph{Blossom-Contracted Graphs} Given $\Omega$, let $G / \Omega$ denote the unweighted simple graph obtained by contracting all the root blossoms in $\Omega$. 
A vertex in $G/\Omega$ is {\it free} if the vertices it represents in $G$ contain a (unique) free vertex. 
The following lemma guarantees that the contraction of the blossoms does not tuck away all augmenting paths.
\begin{lemma}(\cite[Lemma 2.1]{DuanP14}) Let $\Omega$ be a set of full blossoms with respect to a matching $M$.
\begin{enumerate}[leftmargin=*,itemsep=-1ex, topsep = 0pt,partopsep=1ex,parsep=1ex]
  \item If $M$ is a matching in $G$, then $M / \Omega$ is a matching in $G/ \Omega$.
  \item Every augmenting path $P'$ with respect to $M/\Omega$ in $G /\Omega$ extends to an augmenting path $P$  with respect to $M$ in $G$.
  \item Let $P'$ and $P$ be as in (2). Then $\Omega$ remains a valid set of full blossoms after the augmentation $M\gets M\oplus P$.
\end{enumerate}
\end{lemma}

\begin{definition}\label{def:mapping-vertices}
Let $v$ be a vertex in $G/\Omega$, we use $\hat{v}$ to denote the set of vertices in $G$ that contract to $v$. Also, given a set of vertices $S$, define $\hat{S} = \bigcup_{v \in S} \hat{v}$.
For a free vertex $f$ in $G/\Omega$, we define $\dot f$ to be the unique free vertex in $\hat f$.
Given a matched edge $e\in M/\Omega$, we use $\hat{e}$ to denote its corresponding matched edge in $M$.

Conversely, given a set of vertices $S \in \Omega$, let $v^{\Omega}_{S}$ be the vertex in $G/\Omega$ obtained by contracting $S$ in $G$. Given a free vertex in $G$, let $f^{\Omega}$ denote the unique free vertex in $G/\Omega$ that contains $f$. Given a set of free vertices $F$ of $G$, define $F^{\Omega} = \{f^{\Omega} \mid f \in F \}$. Similarly, given a matched edge $e\in M$, if both endpoints belong to different blossoms in $\Omega$, then we define $e^\Omega$ to be the corresponding matched edge in $M/\Omega$.
\end{definition}

\begin{definition}
  Let $H$ be a subgraph of $G$ with a matching $M$.
We denote the set of free vertices in $H$ by $F(H)$ and the set of matched edges in $H$ by $M(H)$.
\end{definition}

\begin{definition}[Inner, outer, and reachable vertices]
Let $F$ be a set of free vertices in a graph $H$ with matching $M$. Let $V^{H,M}_{in}(F)$ and $V^{H,M}_{out}(F)$ denote the set of vertices that are reachable from $F$ with odd-length augmenting paths and even-length augmenting paths respectively. Define $R^{H,M}(F) = V^{H,M}_{in}(F) \cup V^{H,M}_{out}(F)$. When the reference to $H$ and $M$ are clear, we will omit the superscripts and write $R(F)$, $V_{in}(F)$, and $V_{out}(F)$ respectively.
\end{definition}

Notice that using \cref{def:mapping-vertices}, we have $\hat{V}_{in}(F)=\bigcup_{v\in V_{in}(F)}\hat{v}$ and $\hat{V}_{out}(F)=\bigcup_{v\in V_{out}(F)}\hat{v}$.

\subsection{Assumptions to Edge Weights and Approximate Ratio}

Since we are looking for a $(1-\epsilon)$-approximation, we can always re-scale the edge weights to be $O(n/\epsilon)$ while introducing at most $(1-\Theta(\epsilon))$ error (see \cite[Section 2]{DuanP14}). Therefore, we can assume that $\epsilon > 1/n^2$ and so $W \leq n^3$ and $O(\log W) = O(\log n)$; for otherwise we may aggregate the whole network at a node in $O(1 /\epsilon) = O(n^2)$ rounds and have it compute a \MWM{} locally. Let $\epsilon' = \Theta(\epsilon)$ be a parameter that we will choose later. 
We also assume without loss of generality that both $W$ and $\epsilon'$ are powers of two.

\subsection{Assumption of $O((1/\epsilon)\log^3n)$ Weak Diameter}

To begin, we process our input graph by applying a diameter reduction theorem developed by \cite{FFK21} to claim that we may assume that the graph we are considering has a broadcast tree of depth $O((1/\epsilon)\log^3n)$ that can be used to aggregate and propagate information.

\begin{theorem}[\cite{FFK21}, Theorem 7]\label{thm:diameter} Let $T^{\alpha}_{\textsf{SC}}(n,D)$ be the time required for computing an $\alpha$-approximation for the \MWM{} problem in the \suportedcongest model with a communication graph of diameter $D$. Then, for every $\epsilon \in (0,1]$, there is a $\poly(\log n, 1/\epsilon) + O(\log n \cdot T^{\alpha}_{\textsf{SC}}(n,O((1/\epsilon)\log^3 n)))$-round \congest algorithm to compute a $(1-\epsilon)\alpha$-approximation of \MWM{} in the \congest model. If the given \suportedcongest model algorithm is deterministic, then the resulting \congest model algorithm is also deterministic.
\end{theorem}

The \suportedcongest model is the same as the \congest model except that the input graph can be a subgraph of the communication graph.
%Suppose the input graph is a weighted graph $G = (V,E,\hat{w})$. %where $\hat{w}:E \to \{1 ,\ldots, W\}$.
The above theorem implies that we can focus on solving the problem on $G$ as if we were in the \congest model, except that we have access to a broadcast tree (potentially outside $G$) where an aggregation or a broadcast takes $O((1/\epsilon)\log^3n)$ rounds. We slightly abuse the notation and say that $G$ has a {\it weak diameter} of $O((1/\epsilon)\log^3n)$.

With~\cref{thm:diameter}, 
we may broadcast $W$ (the upper bound on edge weights) to every node in $O((1/\epsilon)\log^3n)$ rounds.
We remark that the assumption to the weak diameter is required not only in our algorithm but also in the $(1-\epsilon)$-\MCM{} \congest algorithm described in \cite{FMU22}\footnote{The application of \cref{thm:diameter} can also tie up loose ends left in \cite{FMU22}, where they presented a semi-streaming algorithm first and then described the adaption to other models. One of the primitives, \textsc{Storage} in item (v) in Section 6 assumed a memory of $\Omega(n \poly 1/\epsilon)$ is available to all nodes.  This may be needed in some of their procedures, e.g.~counting $|M_H|$ in Algorithm 7.  The running time was not analyzed, but it may take $O(\mathit{diameter})$ rounds to implement in the \congest model.}.

\subsection{Duan and Pettie's Scaling Framework}

The scaling framework for solving \MWM{} using the primal-dual approach was originally proposed by Gabow and Tarjan~\cite{GT91}.
Let $L=\lfloor \log_2W\rfloor$.
A typical algorithm under this scaling framework consists of $L+1$ scales.
In each scale $i$, such an algorithm puts its attention to the graph with \emph{truncated weights} (whose definition varies in different algorithms).
As $i$ increases, these truncated weights typically move toward the actual input weight.

Duan and Pettie~\cite{DuanP14} introduced a scaling algorithm to solve the $(1-\epsilon)$-\MWM{} problem.
They
proposed a new \emph{relaxed complementary slackness criterion} (see \Cref{prop:dp-relaxed-complementary-slackness}).
The criterion changes between iterations. At the end of the algorithm, the criterion can be used to certify the desired approximation guarantee of the maintained solution.
Unlike Gabow and Tarjan's framework~\cite{GT91}, Duan and Pettie's framework~\cite{DuanP14} allows the matching found in the previous scale to be carried over to the next scale without violating the feasibility, thereby improving the efficiency.
In order to obtain this carry-over feature, Duan and Pettie also introduce the \emph{type $j$ edges} in their complementary slackness criterion.

\begin{definition}[Type $j$ Edges]
A matched edge or a blossom edge is of \emph{type  $j$} if it was last made a matched edge or a blossom edge in scale $j$. 
\end{definition}

Let $\delta_0 = \epsilon'W$ and $\delta_i = \delta_0/2^i$ for all $i\in [0, L]$.
At each scale $i$, 
the \emph{truncated weight} of an edge $e$ is defined as $w_i(e)=\delta_i\lfloor \hat{w}(e)/\delta_i\rfloor$.
The relaxed complementary slackness criteria are based on the truncated weight at each scale.

\begin{lemma}[Relaxed Feasibility and Complementary Slackness {\cite[Property 3.1]{DuanP14}}]\label{prop:dp-relaxed-complementary-slackness}
After each iteration $i=[0, L]$, the algorithm explicitly maintains the set of currently matched edges $M$, the dual variables $y(u)$ and $z(B)$, and the set of active blossoms $\Omega\subseteq \mathcal{V}_{odd}$.
The following properties are guaranteed:
\begin{enumerate}[itemsep=0pt]
    \item {\bf{Granularity.}} For all $B\in\mathcal{V}_{odd}$, $z(B)$ is a nonnegative multiple of $\delta_i$. For all $u\in V(G)$, $y(u)$ is a multiple of $\delta_i/2$. 
    \item {\bf{Active Blossoms.}} $\Omega$ contains all $B$ with $z(B)>0$ and all root blossoms $B$ have $z(B)>0$.
    \item {\bf{Near Domination.}} For all $e\in E$, $yz(e)\ge w_i(e)-\delta_i$.
    \item {\bf{Near Tightness.}} If $e$ is a type $j$ edge, then $yz(e)\le w_i(e)+2(\delta_j-\delta_i)$.
    \item {\bf{Free Vertex Duals.}} If $u\in F(G)$ and $v\notin F(G)$ then $y(u)\le y(v)$.
\end{enumerate}
\end{lemma}

\paragraph{Eligible Graph} 
To achieve~\Cref{prop:dp-relaxed-complementary-slackness} efficiently, at each scale an \emph{eligible graph} $G_{elig}$ is defined. An edge $e$ is said to be \emph{eligible}, if (1) $e\in E_B$ for some $B\in\Omega$, (2) $e\notin M$ and $yz(e)=w(e)-\delta_i$, or (3) $e\in M$ and $yz(e)-w_i(e)$ is a nonnegative integer multiple of $\delta_i$.
$G_{elig}$ is the graph that consists of all edges that are currently eligible.

The algorithm initializes with an empty matching $M\gets\emptyset$, an empty set of active blossoms $\Omega\gets \emptyset$, and high vertex duals $y(u)\gets W/2-\delta_0/2$ for all $u\in V$.
Then, in each scale $i=0, 1, \ldots, L$ the algorithm repeatedly searches for a maximal set $\Psi$ of vertex disjoint augmenting paths in $G_{elig}$, augments these paths, searches for new blossoms, adjust dual variables, and dissolves zero-valued inactive blossoms.
These steps are iteratively applied for $O(1/\epsilon')$ times until the free vertex duals $y(v)$ reach $W/2^{i+2}-\delta_i/2$ whenever $i\in [0, L)$ or $0$ whenever $i=L$.
At the end of scale $L$, \Cref{prop:dp-relaxed-complementary-slackness} guarantees a matching with the desired approximate ratio.
We emphasize that the correctness of~\Cref{prop:dp-relaxed-complementary-slackness} relies on the fact that $\Psi$ is maximal in $G_{elig}$, and the subroutine that searches for $\Psi$ is a modified depth first search from Gabow and Tarjan~\cite{GT91} (see also \cite{MV80,Vazirani94}.)
Unfortunately, some returned augmenting paths in $\Psi$ could be very long.
We do not immediately see an efficient parallel or distributed implementation of this subroutine.

\subsection{Chang and Su's Scaling Framework}

Chang and Su~\cite{CS22} noticed that it is possible to relax Duan and Pettie's framework further,
by introducing the \emph{weight modifiers} $\Delta w(e)$ that satisfy the following new invariants (appended to~\Cref{prop:dp-relaxed-complementary-slackness} with some changes to the other properties) after each iteration: 

{\it{
\begin{enumerate}[itemsep=0pt]
\item[6.] {\bf{Bounded Weight Change.}} The sum of $|\Delta w(e)|$ is at most $\epsilon'\cdot \hat{w}(M^*)$, where $M^*$ is a maximum weight matching with respect to $\hat{w}$.
\end{enumerate}
}}

%These weight modifiers are initially set to $0$ for all edges and will be modified during the iteration.
%By carefully tweaking the definitions to the eligibility of an edge,
Chang and Su showed that 
it is possible to efficiently
obtain a maximal set $\Psi$ of augmenting paths from $G_{elig}$ 
in an expander decomposed $H$-minor-free graph.
By carefully tweaking the definition of the eligibility of an edge,
their modified Duan-Pettie framework fits well into the expander decomposition in the \congest model.
Notice that Chang and Su's scaling algorithm depends on a ``center process'' in each decomposed subgraph.
The center process in each subgraph can obtain the entire subgraph topology within $\polylog(n)$ rounds, with the assumption that the underlying graph is $H$-minor-free.
Thus, a maximal set of augmenting paths in each subgraph can then be computed sequentially in each center process.
This explains two non-trivial difficulties: 
First, it is not clear if the same framework can be generalized to general graphs.
Furthermore, the sequential subroutine searching for augmenting paths may return long augmenting paths. It is not clear how to efficiently implement this subroutine in the PRAM model or the semi-streaming model.

Our scaling framework for general graphs is modified from 
both Duan-Pettie~\cite{DuanP14} and Chang-Su~\cite{CS22}.
In~\cref{sec:scaling} we present our modified scaling framework.
With the adaption of~\cite{FMU22},
we believe our framework is simpler compared with Chang and Su~\cite{CS22}.
Benefiting from~\cite{FMU22} searching for short augmenting paths, the new framework can now be efficiently implemented in the PRAM model  and the semi-streaming model.

\section{Our Modified Scaling Framework}\label{sec:scaling}

Our modified scaling framework maintains the following variables during the execution of the algorithm:

\begin{center}
\begin{tabular}{lll}
$M$:& The set of matched edges. &\\
$y(u)$: & The dual variable defined on each vertex $u \in V$. &  \\
$z(B)$: & The dual variable defined on each $B \in \mathcal{V}_{odd}$. & \\
$\Omega$:& $\Omega \subseteq \mathcal{V}_{odd}$ is the set of \emph{active blossoms}. & \\
$\Delta w (e):$ & The weight modifier defined on each edge $e \in E$. &
\end{tabular}
\end{center}

Our algorithm runs for
 $L+1=\lfloor\log_2 W\rfloor+1$ scales.
In each scale $i$, the same granularity value $\delta_i = \delta_0/2^i$ is used, where $\delta_0 := \epsilon'W$.
Moreover, the truncated weight $w_i(e)$ is now derived from the \emph{effective weight} $w(e) := \hat{w}(e)+\Delta w(e)$, namely $w_i(e)=\delta_i \lfloor w(e)/\delta_i\rfloor$.
There will be $O(1/\epsilon')$ iterations within each scale.
Within each iteration, the algorithm subsequently performs augmentations, updates to weight modifiers, dual adjustments, and updates to active blossoms (see \Cref{sec:iterations-within-each-scale}).

Similar to Chang and Su's framework~\cite{CS22} but opposed to Duan and Pettie's framework~\cite{DuanP14},
the $y$-values of free vertices are no longer the same during the execution.
To make sure that there are still $O(1/\epsilon')$ iterations in each scale,
a special quantity $\tau$ is introduced.
Within each scale $i$,
the quantity $\tau$ will be decreased from $W/2^{i+1}-\delta_{i+1}/2$ to a specified target value $W/2^{i+2}-\delta_{i}/2$ (or $0$ if $i=L$).
Intuitively, 
$\tau$ is the desired free vertex dual which gets decreased by $\delta_i/2$ after every dual adjustment as in \cite{DuanP14} (hence $O(1/\epsilon')$ iterations per scale).
However, in both \cite{CS22} and our framework, 
some free vertices will be isolated from the eligible graph.
This isolation is achieved by
increasing the $y$-value of the free vertex by $\delta_i$.
Therefore, $\tau$ can be seen as a lower bound to all free vertex duals.
Our modified relaxed complementary slackness  (\Cref{prop:RCS}) guarantees that the sum of such gaps will be small.

\paragraph{Eligible Graph} 
The eligible graph $G_{elig}$ is an unweighted subgraph of $G$ defined dynamically throughout the algorithm execution.
The edges in $G_{elig}$ are \emph{eligible edges}.
Conceptually, eligible edges are ``tight'', which are either blossom edges or the ones that nearly violate the complementary slackness condition.
The precise definition of such eligible edges is given in \Cref{def:eligible}.
  
\begin{definition}\label{def:eligible} At scale $i$, an edge is {\it eligible} if at least one of the following holds.
\begin{enumerate}[topsep=0.5ex,itemsep=-.5ex]
  \item $e \in E_B$ for some $B \in \Omega$.
  \item $e \notin M$ and $yz(e) = w_i(e) - \delta_i$.
  \item \label{elig:3} $e \in M$ is a type $j$ edge and $yz(e) = w_i(e) + 2(\delta_j - \delta_i)$.
\end{enumerate}
\end{definition}

We remark that (\ref{elig:3}) is more constrained compared to the Duan-Pettie framework~\cite{DuanP14}.
With the new definition of (\ref{elig:3}), 
an eligible edge can be made ineligible by adjusting its weight modifier $\Delta w(e)$.
Now we describe the relaxed complementary slackness properties:

\begin{property} (Relaxed Complementary Slackness)\label{prop:RCS}
    \begin{enumerate}[topsep=0.5ex,itemsep=-.2ex]
      \item \label{RCS:1} {\bf Granularity.} $z(B)$ and $w_i(e)$ are non-negative multiples of $\delta_i$ for all $B\in \mathcal{V}_{odd}, e \in E$ and $y(u)$ is a non-negative multiple of $\delta_i/2$ for all $u \in V$.
      \item {\bf Active Blossoms.}  \label{RCS:2} All blossoms in $\Omega$ are full.
      If $B \in \Omega$ is a root blossom then $z(B)>0$; if $B \notin \Omega$ then $z(B) = 0$.  Non-root active blossoms may have zero $z$-values.
      \item {\bf Near Domination.} \label{RCS:3} For all edges $e\in E$, $yz(e)\geq w_i(e)-\delta_i$. 
      \item \label{RCS:4}{\bf Near Tightness.} If $e$ is a type $j$ edge, then $yz(e) \leq w_i(e) + 2(\delta_j - \delta_i)$.
      \item \label{RCS:5}{\bf Accumulated Free Vertex Duals.} The $y$-values of all free vertices have the same parity as multiples of $\delta_i / 2$. Moreover, $\sum_{v\in F(G)} y(v) \le \tau \cdot |F(G)| + \epsilon' \cdot \hat{w}(M^{*})$, where $M^{*}$ is a maximum weight matching w.r.t.~$\hat{w}$ and $\tau$ is a variable where $y(v) \geq \tau$ for every $v$.
      %Also, $\tau=0$ if $i=L$.
      Also, $\tau$ decreases to $0$ when the algorithm ends. 
       \item \label{item:boundedw} {\bf Bounded Weight Change.} The weight modifiers are nonnegative. Moreover, the sum of $\Delta w(e)$ is at most $\epsilon' \cdot \hat{w}(M^{*})$.
    \end{enumerate}
    \end{property}

% The main modifications from \cite{DuanP14} and \cite{CS22} are the following: 
% \begin{itemize}[leftmargin=*]
% \item We modified \cref{RCS:5} so that the $y$-values of the free vertices is no longer required to be equal but have the same parity as multiples of $\delta_i / 2$. This is because we may freeze a small fraction of free vertices to prevent their $y$-values from being decreased during an iteration. As a result, they are no longer required to be zero in the end.
% However, the sum of the $y$-values will be upper bounded in the end.
% \item All weight modifiers will be nonnegative in our scaling algorithm, so in \cref{item:boundedw} the absolute function is no longer needed.
% \end{itemize}

With the modified relaxed complementary slackness properties, the following~\Cref{lem:approximateMWM} guarantees the desired approximate ratio of the matching at the end of the algorithm.
As the proof technique is similar to~\cite{DuanP14} and \cite{CS22}, we defer the proof of \Cref{lem:approximateMWM} to \Cref{apx:DPanalysis}.

\begin{lemma}\label{lem:approximateMWM}
Suppose that $y, z, M, \Omega$, and $\Delta w$ satisfy the relaxed complementary slackness condition (\cref{prop:RCS}) at the end of scale $L$. Then $\hat{w}(M) \geq (1-\epsilon) \cdot \hat{w}(M^{*})$.
\end{lemma}

\begin{figure}[h!b]
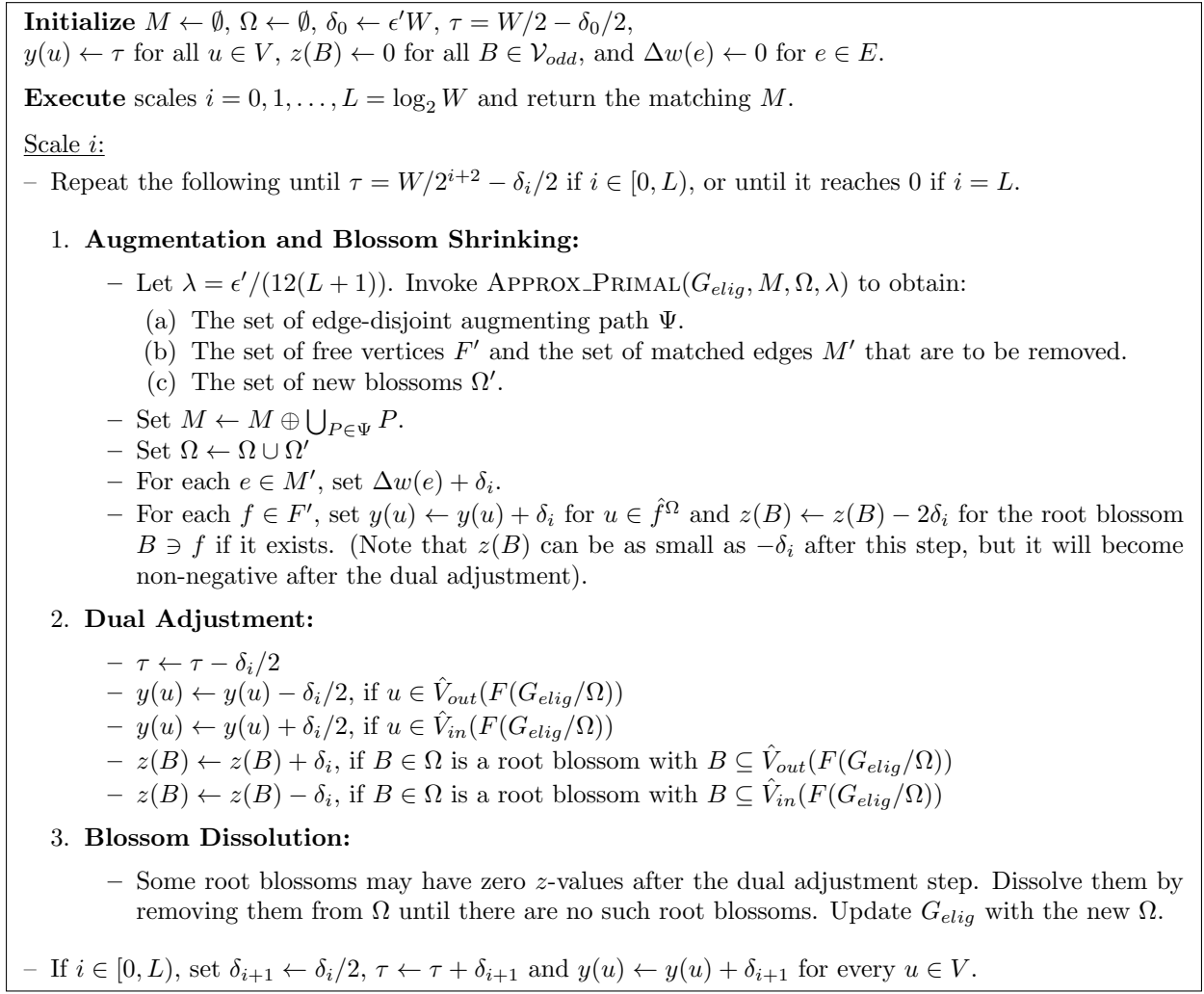

  \centering
  \framebox{
  \begin{minipage}{0.97\textwidth}
  \small
  {\textbf{Initialize}} $M \leftarrow \emptyset$, $\Omega \leftarrow \emptyset$, $\delta_0 \leftarrow \epsilon' W$, $\tau = W/2 - \delta_0/2$,\\
  $y(u) \leftarrow \tau$ for all $u \in V$,
  $z(B) \leftarrow 0$ for all $B \in \mathcal{V}_{odd}$, 
  and $\Delta w(e) \leftarrow 0$ for $e \in E$.
  
  \medskip 
  
  {\textbf{Execute}} scales $i = 0, 1, \ldots, L=\log_{2} W$ and return the matching $M$.
  
  \medskip
  
  \underline{Scale $i$:}
  \begin{itemize}[leftmargin=*]
  \item[--] Repeat the following until $\tau = W/2^{i+2} - \delta_i / 2$ if $i \in [0,L)$, or until it reaches $0$ if $i = L$.
  
  \begin{enumerate}[itemsep=0pt, leftmargin=*]
        
  \item {\bf Augmentation and Blossom Shrinking:} \label{step:2} \label{step:start}
  \begin{itemize}[itemsep=-3pt]
  \item Let $\lambda = \epsilon'/(12(L+1))$. Invoke \textsc{Approx\_Primal}$(G_{elig},M,\Omega, \lambda)$ to obtain:
  \vspace*{-3pt}
  \begin{enumerate}[leftmargin=*,itemsep=-3pt]
    \item The set of edge-disjoint augmenting path $\Psi$.
    
    \item The set of free vertices $F'$ and the set of matched edges $M'$ that are to be removed.
    
    \item The set of new blossoms $\Omega'$.
    \end{enumerate} 
    
  \item Set $M \leftarrow M \oplus \bigcup_{P \in \Psi} P$.  
  \item Set $\Omega \leftarrow  \Omega \cup \Omega'$
  
  \item  For each $e \in M'$, set $\Delta w(e) + \delta_i$. 
  \item For each $f \in F'$, set $y(u) \leftarrow y(u) + \delta_i$ for  $u \in \hat{f}^{\Omega}$ and $z(B) \leftarrow z(B) - 2\delta_i $ for the root blossom $B \ni f$ if it exists. (Note that $z(B)$ can be as small as $-\delta_i$ after this step, but it will become non-negative after the dual adjustment).

  \end{itemize}
  
  \item {\bf Dual Adjustment:} 
  \begin{itemize}[itemsep=-3pt]
    \item $\tau \leftarrow \tau - \delta_i / 2$
      \item $y(u) \leftarrow y(u) - \delta_i / 2$, if $u \in \hat{V}_{out}(F(G_{elig}/\Omega))$
      \item $y(u) \leftarrow y(u) + \delta_i / 2$, if $u \in \hat{V}_{in}(F(G_{elig}/\Omega))$
      \item $z(B) \leftarrow z(B) + \delta_i$, if $B \in \Omega$ is a root blossom with $B \subseteq \hat{V}_{out}(F(G_{elig}/\Omega))$
      \item $z(B) \leftarrow z(B) - \delta_i$, if $B \in \Omega$ is a root blossom with $B \subseteq \hat{V}_{in}(F(G_{elig}/\Omega))$
  \end{itemize}
  \item {\bf Blossom Dissolution:} \label{step:end}
  \begin{itemize}
      \item Some root blossoms may have zero $z$-values after the dual adjustment step. Dissolve them by removing them from $\Omega$ until there are no such root blossoms. Update $G_{elig}$ with the new $\Omega$.
  \end{itemize}
  \end{enumerate}
  \item[--] If $i \in [0,L)$, set $\delta_{i+1} \leftarrow \delta_{i} / 2$, $\tau \leftarrow \tau + \delta_{i+1}$ and $y(u) \leftarrow y(u)+ \delta_{i+1}$ for every $u \in V$.
  \end{itemize}
  \end{minipage}
  }
  \caption{The modified scaling framework.}\vspace*{-1em}
  \label{fig:edmondssearch}
  \end{figure}

\subsection{Iterations in each Scale}
\label{sec:iterations-within-each-scale}

There will be $O(1/\epsilon')$ iterations within each scale.
In each scale $i$, the ultimate goal of the algorithm is to make progress on primal $(M)$ and dual $(y, z)$ solutions such that they meet the complementary slackness properties (\Cref{prop:RCS}). 
This can be achieved by iteratively seeking for a set of augmenting paths $\Psi$, updating the matching $M\gets M\oplus\cup_{P\in\Psi}P$, and then performing dual adjustments on $y$ and $z$ variables.
However, in order to ensure that dual variables are adjusted properly,
we enforce the following \emph{primal blocking conditions} for $\Psi$:

\begin{property}[{\it Primal Blocking Conditions}]\label{prop:PBC}~
\begin{enumerate}[(1),itemsep=0pt]
  \item \label{item:first_block} No augmenting paths exist in $G_{elig}/\Omega$.
  \item  \label{item:second_block} No full blossoms can be reached from any free vertices in $G_{elig}/\Omega$ via alternating paths.
\end{enumerate} 
\end{property}

Here we briefly explain why \Cref{prop:PBC} leads to satisfactory dual adjustments.
In a dual adjustment step, 
the algorithm decreases the $y$-values of inner vertices in $\hat{V}_{in}(F(G_{elig}/\Omega))$ by $\delta_i/2$ and increases the $y$-values of outer vertices in $\hat{V}_{out}(F(G_{elig}/\Omega))$ by $\delta_i/2$.
\Cref{prop:PBC} ensures that $\hat{V}_{in}(F(G_{elig}/\Omega)) \cap \hat{V}_{out}(F(G_{elig}/\Omega)) = \emptyset$ and so the duals can be adjusted without ambiguity.
 
As mentioned in \cref{sec:tech_sum}, it is difficult to 
achieve the primal blocking conditions efficiently in \congest and PRAM due to long augmenting paths and large blossoms.
Fortunately, with the weight modifiers $\Delta w$ introduced from~\cite{CS22},
now we are allowed to remove some matched edges and free vertices from $G_{elig} / \Omega$, which enables a trick of \emph{retrospective eligibility modification}:
after \emph{some} set of augmenting paths $\Psi$ is returned, we modify $G_{elig}$ such that $\Psi$ satisfies \Cref{prop:PBC}.

To remove a matched edge $e$ from $G_{elig}$, we simply add $\delta_i$ to $\Delta w(e)$ and so $e$ becomes ineligible.
To remove a free vertex $f$, we add $\delta_i$ to the $y$-values of vertices in $\hat{f^\Omega}$ and decrease $z(B)$ by $2\delta_i$ where $B$ is the root blossom containing $f$.
By doing so, the vertex $f^{\Omega}$ is isolated from all the other vertices in $G_{elig}/\Omega$ since all the edges incident to $f^{\Omega}$ become ineligible (note that all these edges must be unmatched).  Additionally, 
all the internal edges inside $\hat{f}^{\Omega}$ will have their $yz$-values unchanged.
Note that the reason that we increase the $y$-values by $\delta_i$ instead of $\delta_i/2$ is that we need to synchronize the parity of the $y$-values (as a multiple of $\delta_i/2$) as a technicality required for the analysis.

We present the details of the entire scaling algorithm in \Cref{fig:edmondssearch}.
The {\textbf{augmentation and blossom shrinking step}} is the step that adjusts the primal variables $M$ (and also $\Omega$) and removes some matched edges and free vertices to achieve the primal blocking conditions.
It uses procedure \textsc{Approx\_Primal}, which we describe in \cref{sec:MWMcongest}, that runs in $\poly(1/\epsilon, \log n)$ rounds and returns a set of matched edges $M'$ and free vertices $F'$ of sizes $O((\epsilon / \log n)\cdot |M|)$ as well as a set of augmenting paths $\Psi$ and a set of blossoms $\Omega'$ such that the primal blocking conditions hold in $(G_{elig} - F' - M' - \Psi) / (\Omega \cup \Omega') $.
Assuming such a procedure exists, we give a full analysis in \cref{apx:DPanalysis} to show that the algorithm runs in $\poly(1/\epsilon, \log n)$ rounds and outputs a $(1-\epsilon)$-\MWM{}.

\paragraph{Implementation in {\textsf{CONGEST} model}}
In \congest model, all quantities $M, y(u), y(B), \Omega$, and $\Delta w(e)$ shall be stored and accessed locally.
We present one straightforward implementation in \cref{appendix:data-structures}.
%We describe how a \congest algorithm stores these variables in \cref{appendix:data-structures} in a straightforward way (with a multiplicative slowdown of the maximum active blossom size).
We remark that there is no need to store $\tau$ as a variable since the number of iterations per scale can be pre-computed at the beginning of the algorithm.

\paragraph{Implementation in PRAM and semi-streaming models}
We can simulate the \congest implementation mentioned above in the CREW PRAM model and the semi-streaming model.
The CREW PRAM implementation of 
\textsc{Approx\_Primal} will be described in \cref{sec:pram-approx-primal}.
The semi-streaming implementation will be described in \cref{sec:semi-streaming-impl}.

\section{The Parallel Depth-First Search}\label{sec:FMU}
Fischer, Mitrovic, and Uitto~\cite{FMU22} give a deterministic algorithm for $(1-\epsilon)$-approximate \MCM{} in the semi-streaming model as well as in other models such as \congest and the massive parallel computation (MPC) model.
The core of their algorithm is the procedure \textsc{Alg-Phase}, which searches for an almost maximal set of (short) augmenting paths.
In particular, 
\textsc{Alg-Phase} runs a parallel DFS from every free vertex and returns a set of augmenting paths $\Psi$ and two small-sized sets of vertices $V'$ and $V_A$ such that there exists no augmenting paths of length $O(1/\epsilon)$ on $G-\Psi-V'-V_A$.
The DFS originating from a free vertex $\alpha$ defines a search \emph{structure}, denoted as $S_\alpha$.
For the efficiency purpose, the algorithm imposes several restrictions to the DFS on these structures which are parametrized by a \emph{pass limit} $\tau_{\max}$, a \emph{size limit} $\limit$, and a \emph{depth limit} $\Lmax$.
We provide an overview of the \textsc{Alg-Phase} algorithm of~\cite{FMU22} in~\cref{sec:fmu-overview}.

Unfortunately, directly running \AlgPhase
on $G$ for an almost maximal set of augmenting paths without considering the blossoms would break the scaling framework.
For example, the framework does not allow the search to return two disjoint augmenting paths that pass through the same blossom.
We now describe the modified \AlgPhase that works for the contracted graph $G/\Omega$.

\subsection{{\normalfont\VertexAlgPhase} on the Contracted Graph $G/\Omega$}\label{sec:our-fmu}

Our goal of the modified \AlgPhase is clear: all we need to do is to come up with an almost maximal set $\Psi^\Omega$ of short augmenting paths on $G/\Omega$.
After $\Psi^\Omega$ is found, the algorithm recovers $\Psi$, the actual corresponding augmenting paths on $G$.
Moreover, the algorithm returns two small sets of vertices $V'$ and $V_A$ so that no short augmenting path can be found in $(G-\Psi-V'-V_A)/\Omega$.

Observe that
the lengths of the paths in $\Psi$ on $G$ could be much longer than the paths in $\Psi^\Omega$ on $G/\Omega$, due to the size of blossoms in $\Omega$.
This observation motivates us to consider a \emph{vertex-weighted} version of \textsc{Alg-Phase}.
When computing lengths to an augmenting path on $G/\Omega$, each contracted root blossom now has a weight corresponding to the number of matched edges inside the blossom.
Now we define the \emph{matching length} and \emph{matching distance} in a contracted graph.

\begin{definition}
	Given $u \in G/\Omega$, define $\|u\|=|\hat{u}|$ to be the number of vertices represented by $u$ in the original graph $G$. Given a set of vertices $S$, define $\|S\| = \sum_{u\in S} \|u\|$.\end{definition}

\begin{definition}\label{def:new_matching_length}
Let $M$ be a matching on $G$ and $\Omega$ be a set of full blossoms with respect to $M$ on $G$.
Define $\MContract := M/\Omega$ to be the set of corresponding matched edges of $M$ on $G/\Omega$.
Given an alternating path $P = (u_1, \ldots, u_k)$ in $G / \Omega$,  define the matching length of $P$, $\|P\|_{M} = |P \cap \MContract| + \sum_{i=1}^{k} (\|u_i\|-1)/2$.
For any matched edge $e=uv\in \MContract$ we define $\|e\|_M=(\|u\|+\|v\|)/2$ which corresponds to the total number of matched edges in the blossoms $\hat{u}$ and $\hat{v}$ as well as the edge $e$ itself.
\end{definition}

In the DFS algorithm searching for augmenting paths, a search process may visit a matched edge $e$ in both directions. We distinguish these two situations by giving an orientation to $e$, denoting them as \emph{matched arcs} $\vec{e}$ and $\cev{e}$.
\cref{def:new_matching_length} gives a natural generalization of matching distances on $G/\Omega$:

\begin{definition}\label{def:new_distance} Given a subgraph $H \subseteq G/\Omega$, a set of free vertices $F$, a matching $M$, and a matched arc $\vec{e}$, the matching distance to $\vec{e}$, $d_{H,M}(F,\vec{e})$, is defined to be the shortest matching length of all alternating paths in $H$ that start from a free vertex $F$ and end at $\vec{e}$.  When the first parameter is omitted, $d_{H,M}(\vec{e})$ is the shortest matching length among all alternating paths in $H$ that start from any free vertex in $H$ and end at $\vec{e}$. 
\end{definition}

Throughout this paper,
if an alternating path starts with a free vertex $u_0$ but ends at a non-free vertex, we conveniently denote this alternating path by $(u_0, \vec{e_1}, \vec{e_2}, \ldots, \vec{e_t})$, where $u_0$ is the starting free vertex and $e_1, e_2, \ldots, e_t$ is the sequence of the matched edges along the path.
Also for convenience we define $\|\vec{e}\|_M=\|e\|_M$ for each matched arc $\vec{e}$.

Let $\lambda$ be a parameter\footnote{For the purpose of fitting this subroutine into the scaling framework shown in~\Cref{fig:edmondssearch} and not to be confused with the already-defined parameter $\epsilon$,
we introduce
the parameter $\lambda$ for the error ratio.}.
Similar to the \AlgPhase algorithm, our \VertexAlgPhase returns a collection of disjoint augmenting paths $\mathcal{P}$ where each augmenting path has a matching length at most $O(\poly(1/\lambda))$; two sets of vertices to be removed $V'$ and $V_A$ with their total weight bounded by $\|V'\|=O(\lambda|M|)$ and $\|V_A\|=O(\lambda |M|)$; and the collection of search structures $\mathcal{S}$ where each search structure $S_\alpha\in\mathcal{S}$ has weight $\|S_\alpha\|=O(\poly(1/\lambda))$.
We summarize the vertex-weighted FMU algorithm below:

\begin{lemma}\label{lemma:vertex-weighted-FMU}
  Let $\lambda$ be a parameter.
  Let $G$ be the network with weak diameter $\poly(1/\lambda)$ and $M$ be the current matching.
  Let $\Omega$ be a laminar family of vertex subsets (e.g., the current collection of blossoms) such that 
  each set $B\in\Omega$ contains at most $\Bmax=O(1/\lambda^7)$ vertices.
  Define the DFS parameters $\tau_{\max}:=1/\lambda^4$, $\limit:=1/\lambda^2$, and $\Lmax:=1/\lambda$.
  Then, there exists a \congest and a \crewpram algorithm \VertexAlgPhase such that in $\poly(1/\lambda)$ time, returns $(\mathcal{P}, V', V_A, \mathcal{S})$ that satisfies the following:

\begin{enumerate}[itemsep=0pt]

\item $\|S_{\alpha}\| \leq C_{\max}$ for each structure $S_\alpha\in \mathcal{S}$, \hfill where $\Cmax := \tau_{\max}\cdot (\Lmax+1)\cdot \limit$.

\item $\|V'\| \leq \Cmax\cdot (\Lmax+1)  \cdot (2|\mathcal{P}| + \lambda^{32}\tau_{\max}|M|) $.

\item $\|V_{A}\| \leq h(\lambda)\cdot (2\Lmax) \cdot |M|$, \hfill where $h(\lambda):=\frac{4+2/\lambda}{\lambda\cdot \tau_{\max}} + \frac{2}{\limit}$.

\item No augmenting path $P$ with $\|P\|_M \leq \ell_{\max}$ exists in $(G/\Omega) \setminus (V' \cup V_{A})$.

\item For each matched arc $\vec{e}$, if $d_{(G/\Omega)\setminus (V'\cup V_A), M}(\vec{e}) \le \Lmax$, then there exists a $S_\alpha\in \mathcal{S}$ such that $\vec{e}$ belongs to $S_\alpha$.
\end{enumerate}

Furthermore, \VertexAlgPhase can be simulated in the \crewpram model in $\poly(1/\lambda)$ time with $O(m)$ processors.

\end{lemma}

In \cref{appendix:proof-vertex-weighted-FMU}
we prove \cref{lemma:vertex-weighted-FMU}.
Specifically, we show that 
it is possible to apply a \VertexAlgPhase algorithm on $G/\Omega$ that can be simulated on the underlying network $G$ with an additional $O(\Bmax^2)$ factor in the round complexity.
Interestingly, the \VertexAlgPhase itself is implemented via a black-box reduction back to the unweighted \AlgPhase procedure of \cite{FMU22}.

\section{Augmentation and Blossom Shrinking}
\label{sec:MWMcongest}

The main goal of this section is to prove the following theorem:

\begin{theorem}\label{thm:augment_and_shrink} 
Let $\lambda$ be a parameter.
Given a graph $G$, a matching $M$, a collection of active blossoms $\Omega$ where each active blossom $B\in\Omega$ has size at most $\Cmax=O(1/\lambda^7)$ vertices. There exists a $\poly(1/\lambda, \log n)$-time algorithm in the \congest model that identifies the following:

\begin{enumerate}[(1),itemsep=0pt]
\item\label{enum:main-aug-path-length} A set $\Psi$ of vertex-disjoint augmenting paths with matching lengths at most $\poly(1/\lambda)$.

\item\label{enum:main-small-blossoms} A set of new blossoms $\Omega'$ in $G/\Omega$, where the size of each blossom is at most $\Bmax$.

\item\label{enum:main-not-many-removed-edges}  A set of at most $O(\lambda \cdot |M|)$ matched edges $M'$ and free vertices $F'$. 
\end{enumerate}

\noindent Let $\GContract$ be the contracted graph $\GContract := (G-\Psi-M'-F')/(\Omega\cup \Omega')$ after new blossoms are found and $\FRemain := F(\GContract)$ be the remaining free vertices.
Then, the algorithm also obtains:

\begin{enumerate}[(1),itemsep=0pt]
\setcounter{enumi}{3}

\item\label{enum:main-no-aug-path} \label{enum:main-no-outer-outer}
The sets $\hat{V}_{in}(\FRemain)$ and $\hat{V}_{out}(\FRemain)$.

\end{enumerate}

\noindent These objects are marked locally in the network (see \cref{appendix:data-structures}).
Moreover,
$V_{out}(\FRemain) \cap V_{in}(\FRemain) = \emptyset$. 
\end{theorem}

Notice that the last statement of \Cref{thm:augment_and_shrink} implies that 
 no blossoms can be detected in $\GContract$ from a free vertex, and thus there is no augmenting path from $\FRemain$ on $\GContract$.
 I.e., $(G_{elig}-M'-F')/(\Omega\cup \Omega')$ meets the primal blocking conditions (\Cref{prop:PBC}) after augmenting all the paths in $\Psi$.

 \medskip

\begin{algorithm}[htb]
\caption{\textsc{Approx\_Primal}$(G,M,\Omega, \lambda)$}\label{alg:main}
\begin{algorithmic}[1]\small
\Require Unweighted graph $G$ with weak diameter $O(\frac{\log^3 n}{\epsilon})$, a matching $M$, a collection of blossoms $\Omega$, and $\lambda$ (recall from \Cref{fig:edmondssearch} we set $\lambda=\Theta(\frac{\epsilon'}{\log W})$).
\Ensure Collection of augmenting paths $\Psi$, a set of blossoms $\Omega'$, a set of matched edges $M'$ and a set of free vertices $F'$. These objects as well as $\hat{V}_{in}$ and $\hat{V}_{out}$ are represented locally (see \cref{appendix:data-structures}.)

\State Compute $|M|$.
\vspace*{0.5em}\LineComment{\textbf{Step 1:} Repeatedly search for augmenting paths.}
\Repeat
	 \label{ln:repeat-begin}
		\State $(\mathcal{P}, V', V_{A}, \mathcal{S}) \leftarrow\VertexAlgPhase(G/\Omega, M, \lambda)$ \Comment{See \cref{alg:vertex-weighted-fmu}.}
		\State $\Psi \leftarrow \mathcal{P} \cup \Psi$.
		\State $G\gets G\setminus \mathcal{P}$.
\Until{$|\mathcal{P}|
\leq
\lambda \cdot |M| / (\Cmax(\Lmax+1))  $}\label{ln:repeat-end}

\vspace*{0.5em}
\LineComment{\textbf{Step 2:} Remove matched edges and free vertices returned by the last call to \AlgPhase.}

\State Set $M' \leftarrow M(V') \cup M(V_{A})$ and $F' \leftarrow F(V') \cup F(V_{A})$.\label{ln:collect-active-path}

\vspace*{0.5em}
\LineComment{\textbf{Step 3:} Detect new blossoms.}

\For{each $S_{\alpha} \in \mathcal{S}$}
\label{ln:detect-blossom-for-loop}
			\State Detect all the (possibly nested) blossoms in $S_{\alpha}$ and add them to $\Omega'$.\label{ln:detect-blossom-step}
\EndFor

\vspace*{0.5em}
\LineComment{\textbf{Step 4:} Remove some matched edges and free vertices such that all long augmenting paths disappear.}
\State
Define $\GContract := (G-\Psi-M'-F')/(\Omega\cup \Omega')$, $\FRemain=F(\GContract)$, and $\MContract:=M(\GContract)$.
\State
Use \cref{lemma:simulate-bfs} to compute distance labels $\ell(\vec{e})\in \{d_{\GContract, \MContract}(\FRemain, \vec{e}), \infty \}$
for each matched arc $\vec{e}\in\MContract$.
\label{ln:bfs} 
\State Let $E_i, F_i\gets \emptyset$ for all $i\in 1, 2, \dots, \lfloor\Lmax/2\rfloor$.
\For{each matched arc $\vec{e}\in\MContract$ such that $\ell(\vec{e}) \neq \infty$}
\State Add the corresponding matched edge $\hat{e}\in M$ to $E_i$ for all $i\in [\ell(\vec{e}) - \|e\|_M, \ell(\vec{e})]\cap [1, \Lmax/2]$.
\EndFor
\For{each free vertex $f \in \FRemain$}
\State Add the corresponding free vertex $\dot f\in F$ to $F_i$ for all $i\in [1, (\|f\|-1)/2]$.
\EndFor
\State Let $i^{*} = \arg \min_{1 \leq i \leq \Lmax/2} \{|E_i| + |F_i|\}$.
\label{ln:partition-edges}
\State Add all the matched edges in $E_{i^{*}}$ to $M'$ and all free vertices in $F_{i^*}$ to $F'$.
\label{ln:lastline}
\State\Return $\Psi, \Omega', M'$, and $F'$.

\end{algorithmic}
\end{algorithm}

To prove \cref{thm:augment_and_shrink}, we propose 
the main algorithm \cref{alg:main}.
The algorithm consists of four steps.
In the first step (\cref{ln:repeat-begin} to \cref{ln:repeat-end}), the algorithm repeatedly invokes the \AlgPhase{} procedure and obtains a collection $\mathcal{P}$ of (short) augmenting paths.
These augmenting paths, once found, are temporarily removed from the graph.
The loop ends once the number of newly found augmenting paths becomes no more than $\lambda\cdot |M|/\Cmax$.
Notice that we are able to count $|\mathcal{P}|$ in $\poly(1/\lambda, \log n)$ time 
because $G$ has a weak diameter $O(\log^3 n/\epsilon)$.
The algorithm then utilizes the output $(\mathcal{P}, V', V_A, \mathcal{S})$ from the last execution of $\VertexAlgPhase{}$ in the subsequent steps.

In step two (\cref{ln:collect-active-path}) the algorithm removes edges in $M(V')\cup M(V_A)$ and free vertices in $F(V')\cup F(V_A)$. This ensures that no short augmenting paths can be found from the remaining free vertices.

In the third step (\cref{ln:detect-blossom-for-loop} to \cref{ln:detect-blossom-step}), the algorithm detects and contracts $\Omega'$ ---
the set of all blossoms within any part of $\mathcal{S}$.
We remark that there could still be an edge connecting two outer vertices in $V_{out}(\GContract)$ after the contraction of the blossoms from $\Omega'$, leading to an undetected augmenting path.
However, in this case, we are able to show that at least one endpoint of the edge must be far enough from any remaining free vertex, so
after the fourth step, such an \emph{outer-outer edge} no longer belongs to any augmenting path.

The fourth step (\cref{ln:bfs} to \cref{ln:lastline}) of the algorithm assembles the collection of matched edges $\{E_{i}\}_{i=1,2,\ldots, \floor{\Lmax/2}}$ and free vertices $\{F_i\}_{i=1,2,\ldots, \floor{\Lmax/2}}$.
Each pair of sets $(E_{i}, F_i)$ has the property that after removing the all matched edges in $E_i$ and free vertices in $F_i$, there will be no more far-away outer-outer edges (and thus no augmenting path).
Therefore, to eliminate all far away outer-outer edges, the algorithm chooses the index $i^*$ with the smallest  $|E_{i^*}|+|F_{i^*}|$ and then removes all matched edges in $E_{i^*}$ and free vertices in $F_{i^*}$.
Intuitively, any long enough alternating paths starting from a free vertex will be intercepted at the matching distance $i^*$ by $E_{i^*}$ and $F_{i^*}$.

Let $\GContract$ be the current contracted graph $(G-\Psi-M'-F')/(\Omega\cup \Omega')$ after removing a set of augmenting paths $\Psi$, a set of matched edges $M'$ and a set of free vertices $F'$.
To form the collection of matched edges $\{E_{i}\}$ and free vertices $\{F_i\}$, the algorithm runs a Bellman-Ford style procedure that computes distance labels to each matched arc $\ell(\vec{e})$.
The goal is to obtain $\ell(\vec{e})=d_{\GContract, \MContract}(\FRemain, \vec{e})$ whenever this matching distance is no more than $\Lmax + \|e\|_M$, and $\ell(\vec{e})=\infty$ otherwise.
We note that the labels can be computed efficiently because the (weighted) vicinity of the free vertices, after the blossom contractions, is now bipartite.

For each matched arc $\vec{e}$ with a computed matching length $\ell(\vec{e}) = d_{\GContract, \MContract}(\FRemain, \vec{e})\le \Lmax + \|e\|_M$, we add the corresponding matched edge $\hat{e}\in M$ to $E_{i}$ for all integers $i\in [\ell(\vec{e})-\|e\|_M, \ell(\vec{e})] \cap [1, \Lmax/2]$.
In addition, for each free vertex $f\in \FRemain$, its corresponding free vertex $\dot f\in F$ is added to $F_i$ for all $i\in [1, (\|f\|-1)/2]$. 
Finally, $i^* = \arg\min_{i}\{|E_i|+|F_i|\}$ can be computed and then $E_{i^*}$ and $F_{i^*}$ are removed from $G$.

The rest of the section proves \cref{thm:augment_and_shrink}.

\subsection{Correctness of \cref{thm:augment_and_shrink}}

First of all, we notice that $\Psi$ is comprised of augmenting paths returned by \AlgPhase, which is parametrized to return augmenting paths of matching length at most $2\Cmax = \poly(1/\lambda)$. Thus \cref{enum:main-aug-path-length} of \cref{thm:augment_and_shrink} holds.
Moreover, since in Step 3 (\cref{ln:detect-blossom-step}) the algorithm only searches for blossoms within each part in $\mathcal{S}$, the size of each blossom must be at most $\Cmax$. Thus \cref{enum:main-small-blossoms} holds.

Now, we turn our attention to \cref{enum:main-not-many-removed-edges}.
We notice that the set of removed matched edges $M'$ and free vertices $F'$
 are only affected in \cref{ln:collect-active-path} and \cref{ln:lastline}.
 \cref{lemma:main-bound-removed-edges-1} 
  focuses on \cref{ln:collect-active-path}, and \cref{lemma:size-of-layered-matched-edges} focuses on \cref{ln:lastline}:

\begin{lemma}\label{lemma:main-bound-removed-edges-1}
After the execution of \cref{ln:collect-active-path}, both the size of the sets $M'$ and $F'$ are $O(\lambda\cdot |M|)$. 
\end{lemma}

\begin{proof} 
	It suffices to upper bound the four quantities $|M(V')|$, $|M(V_A)|$, $|F(V')|$, and $|F(V_A)|$ individually.
	Since the repeat loop (\cref{ln:repeat-begin} to \cref{ln:repeat-end}) stops whenever the number of augmenting paths $|\mathcal{P}|$ is upper bounded by $\lambda\cdot |M|/(\Cmax(\Lmax+1))$, we have

	\begin{equation}
		\begin{aligned}
		|M(V')| &\le \|V'\|/2\\
		&\le \Cmax (\Lmax+1) \cdot (|\mathcal{P}|
		+ \lambda^{32} \tau_{\max} |M|/2) \ \ \ \ \ \ \ \ \ \ \ \text{(By property 2 of \cref{lemma:vertex-weighted-FMU})}\\
		&\le \Cmax (\Lmax+1) \cdot \lambda\cdot |M| / (\Cmax(\Lmax+1)) + O(\lambda^{32} \Cmax\Lmax\tau_{\max} |M|) \\
		&= O(\lambda\cdot |M|) \hspace{8em} (\Cmax=O(1/\lambda^7), \Lmax=1/\lambda, \text{ and }\tau_{\max}=1/\lambda^4)
		\end{aligned}
	\end{equation}

In addition, we have $|F(V')| \le \|V'\| = O(\lambda\cdot |M|)$. 
Now, because each active path can be decomposed to exactly one free vertex and several matched edges, we have
\begin{equation}
	\begin{aligned}
		|F(V_A)| + 2|M(V_A)| &= |V_A| \le \|V_A\| \\
		&\le h(\lambda)\cdot (2\Lmax)\cdot |M| & \text{(By property 3 of \cref{lemma:vertex-weighted-FMU})} \\
		&= h(\lambda)\cdot (2/\lambda)\cdot |M| \\
		&= O(\lambda\cdot |M|) & \text{(Notice that $h(x)=O(x^2)$ for all small $x$)}\\
	\end{aligned}
\end{equation}

Therefore, we conclude that
$|M'|= O(\lambda\cdot|M|)$ and
$|F'|= O(\lambda\cdot|M|)$ after \cref{ln:collect-active-path}. 
\end{proof}

Now, we claim that in \cref{ln:lastline}, the algorithm removes at most $O(\lambda\cdot |M|)$ matched edges and free vertices.
The claim is implied by the following \cref{lemma:size-of-layered-matched-edges}, which states that the total size of the collection does not exceed $2|M|$.
Therefore, by the fact that $i^*$ minimizes the total size $|E_{i^*}|+|F_{i^*}|$, \cref{ln:lastline} adds at most $2|M|/\floor{\Lmax/2} = O(\lambda\cdot |M|)$ matched edges and free vertices to $M'$ and $F'$.

\begin{lemma}\label{lemma:size-of-layered-matched-edges}
$\sum_{i=1}^{\floor{\Lmax/2}} |E_{i}| + |F_{i}| \le 2|M|$.
\end{lemma}

\begin{proof}
\begin{align*}
	\sum_{i=1}^{\floor{\Lmax/2}} |E_i|+|F_i| &\le \sum_{e\in \MContract} (\|\vec{e}\|_M+1+\|\cev{e}\|_M+1) + \sum_{f\in \FRemain} (\|f\|-1)/2 \tag{the total size equals to all occurrences of each arc and free vertex in the collection}\\
	&\le 2|\MContract| + \sum_{e=uv\in \MContract} 2\left(\frac{\|u\|-1}{2}+\frac{\|v\|-1}{2}\right) + \sum_{f\in \FRemain} \frac{\|f\|-1}{2} \\
	&\le 2|\MContract| + 2\sum_{v\in \GContract} \frac{\|v\|-1}{2} \tag{every vertex is incident to at most one edge in $\MContract$}\\
	&\le 2|\MContract| + 2(|M|-|\MContract|)\\
	&=2|M| \qedhere
\end{align*}
\end{proof}

To prove \cref{enum:main-no-outer-outer}, we first show that any shortest alternating path with a matching length at least $\Lmax=1/\lambda$ must intersect either $E_{i}$ or $F_{i}$ for any integer $i\in [1, \Lmax/2]$.

\begin{lemma}\label{lem:cut}
Consider an alternating path $P = (v_0, \vec{e}_1, \ldots, \vec{e}_t)$ on $\GContract$, where $v_0\in \FRemain$ is a free vertex and $e_j$ is a matched edge for $1 \leq j \leq t$.
Assume $\ell(\vec{e}_t)=d_{\GContract, \MContract}(\FRemain, \vec{e}_t) \ge \Lmax$.
Then, for any $i\in [1, \Lmax/2]$,
$ \{\dot{v_0}, \hat{e_1}, \cdots, \hat{e}_t\}\cap (E_{i}\cup F_{i})\neq \emptyset$.
 \end{lemma}

\begin{proof}
By definition, $\dot{v_0}$ occurs in $F_i$ for all $i\in [1, (\|v_0\|-1)/2]$ and all matched edges $e_j$ occurs in $E_i$ for all $i\in [\ell(\vec{e}_j) - \|e_j\|_M, \ell(\vec{e}_j)] \cap [1, \Lmax/2]$.
Since $\ell(\vec{e}_t) \neq \infty$, we know that $\ell(\vec{e}_j) = d_{\GContract,\MContract}(\FRemain, \vec{e}_j)$ for all $j=1,2,\ldots, t$.
Moreover, $P$ is an alternating path so by definition of matching length we know that whenever $j>1$ we have $\ell(\vec{e}_j) \le \ell(\vec{e}_{j-1}) + \|e_j\|_M$ and
when $j=1$ we have $\ell(\vec{e}_1) \le (\|v_0\|-1)/2 + \|e_1\|_M$.
Therefore, using the assumption that $\ell(\vec{e}_t)\ge \Lmax$ we obtain 
$$
[1, \Lmax/2]\subseteq [1, (\|v_0\|-1)/2] \cup \bigcup_{j\ge 1} [\ell(\vec{e}_j) - \|e_j\|_M, \ell(\vec{e}_j)].
$$
Thus, for any $i\in [1, \Lmax/2]$ either $\dot{v_0}\in F_i$ or there are some $j$ such that $\hat{e_j}\in E_i$.
\end{proof}

\cref{lem:cut} implies that there is no augmenting path leaving from any of the free vertex $\FRemain$ on $\GContract$.
But it does not imply \cref{enum:main-no-outer-outer} since there could be an edge connecting two outer vertices in $V_{out}(\FRemain)$ without an augmenting path.
The next lemma (\cref{lemma:no-outer-outer}) shows that such a situation does not happen after contracting blossoms (Step 3) and removing the thinnest layer of matched edges and free vertices (Step 4).

\begin{lemma}\label{lemma:no-outer-outer}
	Fix any integer $i\in [1, \Lmax]$.
	Let $H = (G-\Psi-M'-F'- E_i-F_i)/(\Omega\cup\Omega')$.
	%\GContract - E_{i}^{\Omega\cup \Omega'} - F_{i}^{\Omega\cup \Omega'}$.
	Then there is no unmatched edge connecting two vertices in $V_{out}^{H, \MContract}(\FRemain\setminus F_{i}^{\Omega\cup \Omega'})$.
\end{lemma}

\begin{proof}
	Let 
	$uv$ be an unmatched edge on $G$ with $u^{\Omega\cup\Omega'}, v^{\Omega\cup\Omega'}\in V_{out}^{\GContract, \MContract}(\FRemain)$ but $u^{\Omega\cup\Omega'}\neq v^{\Omega\cup\Omega'}$.

	With \cref{lem:cut} in mind, it suffices to prove the following claim: there exists one of the vertices $x\in\{u, v\}$ such that either 
	\begin{enumerate}[itemsep=0pt]
		\item $x^{\Omega\cup \Omega'}\in \FRemain$ and  $(\|x^{\Omega\cup \Omega'}\|-1)/2 \ge \Lmax/2$, or
		\item $x^{\Omega\cup \Omega'}\notin \FRemain$ and the matching distance $d_{\GContract, \MContract}(\FRemain, \vec{e}_x) \ge \Lmax/2$ where $e_x\in M$ is the matched edge incident to $x^{\Omega\cup \Omega'}$ on $\GContract$.
	\end{enumerate}	
	Once we have the above claim for some vertex $x\in\{u, v\}$, we obtain a contradictory argument because now either $x^{\Omega\cup\Omega'}\in F_{i^*}^{\Omega\cup \Omega'}$, or
  there exists a shortest augmenting path from a free vertex in $\FRemain$ to $x^{\Omega\cup\Omega'}$ of matching distance at least $\Lmax$, which is cut off by some set $E_{i^*}^{\Omega\cup \Omega'}$ in the middle by \cref{lem:cut} and thus $x^{\Omega\cup\Omega'}\notin V_{out}^{\GContract, \MContract}(\FRemain)$.

	Now 
	we prove the claim by another contradiction.
	Suppose the statements 1. and 2. in the claim are all false for both $x=u$ and $x=v$.
	Using the fact that the contraction never decreases the matching distances, we know that for \emph{both} $x\in\{u, v\}$,
	either 
	\begin{enumerate}
	\item $x^\Omega\in F^\Omega$ but $(\|x^\Omega\|-1)/2 < \Lmax/2$, or 
	\item $x^\Omega \notin F^\Omega$ but $d_{(G-\Psi-M'-F')/\Omega, M}(F^\Omega, \vec{e}_x) < \Lmax/2$ where $e_x\in M$ is the matched edge incident to $x^\Omega$.
	\end{enumerate}
	
	Now, both $u^\Omega$ and $v^\Omega$  do not belong to any active path from the last execution of \VertexAlgPhase.
	Furthermore, both $u^\Omega$ and $v^\Omega$ must belong to some structure by Item 5 in \cref{lemma:vertex-weighted-FMU}.
	If $u^\Omega$ and $v^\Omega$ belong to different structures, then there must be an augmenting path of matching length $<2(\Lmax/2)=\Lmax$ which contradicts with Item 4 of \cref{lemma:vertex-weighted-FMU}.
	Hence, 
	we conclude that both $u^\Omega$ and $v^\Omega$ belongs to the same structure $S_\alpha\in \mathcal{S}$.
	However, the assumption states that both $u^\Omega$ and $v^\Omega$ are outer vertices in $V_{out}^{(G-\Psi-M'-F')/\Omega, M}(F^\Omega)$.
	Hence, in Step 3 the algorithm creates a blossom that contains both $u^\Omega$ and $v^\Omega$, which implies $u^{\Omega\cup \Omega'} = v^{\Omega\cup\Omega'}$, a contradiction.
\end{proof}

The proof of \cref{enum:main-no-aug-path} in \cref{thm:augment_and_shrink} now follows immediately from \cref{lemma:no-outer-outer}.

\subsection{Implementation Details in \cref{alg:main} in \textsf{CONGEST}}

There are 3 tasks in \cref{alg:main} that are unclear for implementation in the \congest model.
These tasks are (1) obtaining the correct counts of $|\mathcal{P}|$ (\cref{ln:repeat-end}), $|E_i|$ and $|F_i|$ (\cref{ln:partition-edges}),
(2) correctly identifying and forming blossoms within each $\mathcal{S}_\alpha$ (\cref{ln:detect-blossom-step}), and
(3) computing distance labels for matched arcs on $\GContract$ (\cref{ln:bfs}).

Task (1) can be solved in $O(\log n)$ rounds per set using the underlying communication network.

For Task (2), we simulate the naive sequential algorithm for formulating blossoms in each structure $S_\alpha\in \mathcal{S}$:

\begin{lemma}\label{lemma:forming-blossoms}
Let $S_\alpha$ be a structure returned from an execution of \VertexAlgPhase.
Then, there exists an algorithm in \congest that detects (hierarchy of) blossoms within $S_\alpha$ in $\poly(1/\lambda)$ rounds. 
\end{lemma}

\begin{proof}
	Since each $S_\alpha$ has size at most $\Cmax$, it suffices to spend $O(\Cmax^2)$ rounds to aggregate the entire induced subgraph of vertices in $S_\alpha$ to the free node $\alpha$.
	After the node $\alpha$ identifies all blossoms locally, it broadcasts this information to all the vertices within $S_\alpha$ using another $O(\Cmax^2)$ rounds.
\end{proof}

For Task (3), we give a Bellman-Ford style algorithm on $\GContract$ in a straightforward way:

\begin{lemma}\label{lemma:simulate-bfs}
Assume every node knows the blossom size and the root of its associated root blossom (if there is one), and the incident matched edge (if there is one).
Assume that each blossom has its diameter at most $O(\Lmax)$.
Then, after Step 3 is performed, there exists an algorithm in \congest that computes the distance labels $\ell(\vec{e})$ to matched arcs with $\ell(\vec{e}) = d_{\GContract, \MContract}(\FRemain, \vec{e})$ whenever $d_{\GContract, \MContract}(\FRemain, \vec{e})\le \Lmax + \|e\|_M$.
This algorithm finishes
  in $O(\Lmax^2) = \poly(\Lmax)$ rounds.
\end{lemma}

\begin{proof}
We first describe the algorithm.
Initially, using $O(\Lmax)$ rounds, all matched arcs that can be reached by a free vertex $f$ obtain the distance label $\ell(\vec{e})=\|f\|+\|e\|_M$.
For all other matched arcs the label is set to be $\ell(\vec{e})=\infty$.
Then, 
the rest steps of the algorithm are split into $\Lmax$ iterations $t=1, 2, \ldots, \Lmax$.

In each iteration $t$, each matched arc $\vec{uv}$ with a label $\ell(\vec{uv})=t$ informs the neighbors of $v$ about this label.
For each neighbor $x$ of $v$ who receives this information attempts to update the associated matched arc $\vec{xy}$ by setting $\ell(\vec{xy}) \gets \min\{\ell(\vec{xy}), t+\|xy\|_M\}$, similar to a relaxation step in the Bellman-Ford algorithm.

\paragraph{Correctness}
It is straightforward to see that $\ell(\vec{e})\le d_{\GContract, \MContract}(\FRemain)$ whenever $\ell(\vec{e})<\infty$.
Moreover, whenever $\ell(\vec{e}) < \infty$ we must have $\ell(\vec{e})\le \Lmax + \|e\|_M$.

Now we prove that for each matched arc $\vec{e}$ with $\ell(\vec{e})<\infty$,
there exists an alternating path of matching length exactly $\ell(\vec{e})$, so that $\ell(\vec{e})\ge d_{\GContract, \MContract}(\FRemain)$ which implies the equality.

Suppose this is not true, 
then there exists a matched arc $\vec{e}$ such that $\ell(\vec{e}) < d_{\GContract, \MContract}(\FRemain)$.
According to the algorithm, there exists a walk of interleaving unmatched and matched arcs such that the walk ends with $\vec{e}$. Moreover, the sum of all weighted matched arcs in the walk is exactly $\ell(\vec{e})$.
Since $\ell(\vec{e}) < d_{\GContract, \MContract}(\FRemain)$, the walk must \emph{not} be an alternating path.
This implies that a matched edge whose both directional arcs appear in the walk.
Hence, there must exist an unmatched edge connecting two outer vertices must be observed.
However, this contradicts again the claim in the proof of \cref{lemma:no-outer-outer} since a blossom should have been formed during Step 3.

\paragraph{Remark}
It can be shown that if an arc $\vec{e}$ has $\ell(\vec{e}) < \Lmax$ then $\ell(\cev{e})=\infty$. That is, the graph induced by finitely labeled matched edges (as well as those unmatched edges used for relaxation) is very close to a bipartite graph, in the sense that all matched arcs are most likely to be visited from an inner vertex to an outer vertex.
The proof on the bipartite graph becomes trivial as every distance label corresponds to an alternating path on bipartite graphs.
However, there could be some edges ``at the boundary'' where both $\ell(\vec{e})$ and $\ell(\cev{e})$ are within $[\Lmax, \infty)$ so the graph we are concerning is not quite bipartite.

\paragraph{Runtime}
Finally, we analyze the runtime.
Since there are $\Lmax$ iterations and each iteration takes $O(\Lmax)$ rounds to propagate the labels, the total runtime
is $O(\Lmax^2)$.
\end{proof}

\subsection{Runtime Analysis in \cref{thm:augment_and_shrink}}

The following lemma summarizes the runtime analysis to \cref{thm:augment_and_shrink}.

\begin{lemma}\label{lemma:main-alg-runtime}
\cref{alg:main} finishes in $\poly(1/\lambda,\log n)$ rounds.
\end{lemma}

\begin{proof}
In Step 1, each iteration in a repeated loop involves one execution of \VertexAlgPhase and $O((1/\epsilon)\log^3 n)$ additional rounds for counting augmenting paths.
Moreover, there can be at most  $$2+|M| \frac{1}{\lambda\cdot|M|/(\Cmax(\Lmax+1))} =\poly(1/\lambda)$$ iterations.
The reason is,
starting from the second iteration, all augmenting paths found include at least one matched edge, so at least $\lambda |M|/\Cmax$ matched edges will be removed from the graph.
By \cref{lemma:vertex-weighted-FMU}, each \VertexAlgPhase takes $\poly(1/\lambda, \log n)$ rounds (see also~\cref{lemma:simulate-fmu}).
Thus, Step 1 takes $\poly(1/\lambda, \log n)$ rounds.

Step 2 takes $O(1)$ rounds.

Step 3 takes up to $\poly(1/\lambda)$ rounds by \cref{lemma:forming-blossoms}.

Step 4 has two parts, performing a Bellman-Ford style algorithm takes $\poly(1/\lambda)$ rounds by \cref{lemma:simulate-bfs}. In the second part, the algorithm computes the size of each set $|E_i|$ and $|F_i|$, which takes $O(\log n + 1/\lambda)$ rounds.
Therefore, the total runtime to \cref{alg:main} is $\poly(1/\lambda) + O(\log n) = \poly(1/\lambda, \log n)$ as desired.
\end{proof}

\subsection{\Cref{alg:main} in CREW PRAM model}\label{sec:pram-approx-primal}

We simulate the previous \congest implementation of \Cref{alg:main}, so it suffices to show that all local decisions can be done efficiently (with an extra factor of $O(\Cmax^2+\log n)$ parallel time with $O(m)$ processors).
We remark that the \crewpram implementation does not require the assumptions of the $O((1/\epsilon)\log^3 n)$ weak diameter.

\begin{itemize}
	\item Obtaining the matching size $|M|$ and the set size of augmenting paths $|\mathcal{P}|$ can be done in $O(\log n)$ time via the standard parallel prefix sum operation. 
	\item In Step 1, each \VertexAlgPhase takes $\poly(1/\lambda, \log n)$ time using $O(m)$  processors by \cref{lemma:vertex-weighted-FMU} (see also \cref{lemma:simulate-fmu-pram}).
	\item In Step 2, marking matched edges and free vertices takes $O(1)$ time.
	\item In Step 3, identifying new blossoms within each structure $S_\alpha$
	can be done sequentially in $O(\Cmax^2)$ time~\cite{GT91} as each structure contains at most $O(\Cmax^2)$ edges.
	\item In Step 4, the parallel implementation to \Cref{lemma:simulate-bfs} requires relaxing the distances of a matched arc in parallel.
	Each Bellman-Ford step of simultaneously relaxing a set of matched arcs takes $O(\log n)$ time with $O(m)$  processors.
\end{itemize}
Therefore, \Cref{thm:main-parallel} follows immediately.

\section*{Acknowledgements}

The authors would like to thank Aaron Bernstein and Aditi Dudeja for their invaluable discussions on solving $(1-\epsilon)$-MWM in the semi-streaming model, and, in particular, for pointing out the reduction of \cite{GuptaP13} to us. The authors also thank the anonymous reviewers for their helpful comments.

{\small
\bibliographystyle{alpha}
\bibliography{mwm}
}

\clearpage
\appendix

% \section{Tables} 

% \clearpage

\section{Data Structures for Executing Algorithms with Blossoms in {\normalfont\textsf{CONGEST}}}\label{appendix:data-structures}

This section discusses the data structures stored on each node $v\in V$ over the network $G$, as well as the implementation details for operations in a blossom.

\paragraph{Vertices and Edges}
Each vertex has its own ID, the list of its neighbors' ID, and its own $y$-value.
The information of an edge $uv\in E$ is accessible (and stored) via both endpoints $u$ and $v$.
For example, if an edge $uv$ is a matched edge, then both $u$ and $v$ \emph{knows} that $uv\in M$; similarly if $uv$ is an unmatched edge but belongs to some augmenting path $P$, then both $u$ and $v$ knows that $uv\in P$.
In particular, if an arc (or an edge) is labeled with any information, this information is accessible (and stored) in both endpoints. $M$ and $\Delta w(\cdot)$ are stored in this way. The $y$-values are stored on each vertex.

\paragraph{Markers}
A \emph{marker} is a string with an ID stored on a vertex that certifies that this vertex belongs to some object with the corresponding ID.
These objects include the matching $M$, the sets $\hat{V}_{in}$ and $\hat{V}_{out}$, the structures $S_\alpha\in\mathcal{S}$, and the blossoms $B\in\Omega$.
For the sets $V_{in}$ and $V_{out}$, each vertex in the set has a marker with the ID of the set. We now discuss about the blossoms.

\paragraph{Blossoms}
We only need to maintain full blossoms.
Each full blossom $(B, E_B)$ has a unique ID and every vertex $v\in B$ has a marker for the blossom.
Recall from \cref{sec:scaling} that 
for each non-trivial blossom there is an odd number of $2k+1$ full sub-blossoms $B_0, B_1, \ldots, B_{2k}$
and edges $e_0, e_1, \ldots, e_{2k}$ such that 
$B=\bigcup_i B_i$ and $e_i\in B_i\times B_{(i+1)\bmod (2k+1)}$.
Moreover, these edges are alternating: $e_0, e_2, \ldots, e_{2k}\notin M$ but $e_1, e_3, \ldots, e_{2k-1}\in M$.
The \emph{base} $b$ of a blossom is the vertex in $B_0$ that is either a free vertex or has an incident matched edge that does not belong to $E_B$.
A root blossom is a blossom that is not a subset of any other blossom from $\Omega$.
In our algorithm, any blossom maintained in our algorithm never exceeds $\Bmax$ vertices (see~\Cref{thm:augment_and_shrink}\ref{enum:main-small-blossoms}).

\subsection{Lazy Implementation} \label{sec:lazy-implementation}
We use a lazy approach to implement the changes to a blossom.
That is, the base $b$ of a root blossom $B$ controls everything within $B$ and \emph{all non-root blossoms} in $B$.
In particular, $b$ stores the entire induced subgraph $E(B)$, the matched edge sets $E(B)\cap M$, and all dual variables $z(B')$ where $B'\in \Omega$ and $B'\subseteq B$.
A communication tree $T_B$ within $B$ (rooted at $b$) is also used for broadcasting the messages.
Each non-base vertex $v\in B$, $v\neq b$ stores the parent of $v$ in $T_B$ that helps for communication.

Now we show that every operation needed from the algorithm can be implemented in $O(\Bmax^2)$ rounds:

\paragraph{Forming a Blossom}
A blossom in our algorithm can only be formed via Step 3 of \cref{alg:main}.
By \cref{lemma:forming-blossoms} the blossoms are computed and then broadcast from a structure $S_\alpha$.
After all blossoms are identified within $S_\alpha$, the free vertex $\alpha$ sends the sets of blossoms to each root blossom's base.
Finally, an additional $O(\Bmax^2)$ rounds are used for gathering all information about each root blossom within the structure.

\paragraph{Blossom Dissolution}
A blossom can only be dissolved in
 Step 3 of \cref{fig:edmondssearch}.
 Whenever a root blossom dissolves (according to the criteria $z(B)=0$), the base $b$ computes locally and decides which blossoms remain to be root blossoms.
 Then, in $O(\Bmax^2)$ rounds the base $b$ passes the information ($z$-values) to each new root blossoms.

\paragraph{Traversal and DFS}
Another challenge is to support DFS operations on $G/\Omega$.
In particular, a \textsc{Next-Incident-Edge}$(B, uv\text{ or }\perp)$ should return the \emph{next} edge incident to the blossom $B$, where $uv$ is the previously considered incident edge. (If $\perp$ is the argument then the first edge is returned.)
Since we can impose a total order to these edges, this operation can be done in $O(\diamB)$ rounds
by (1) broadcast to every node in $B$ saying that $uv$ is the current edge,
(2) every vertex returns the lexicographically smallest ``outgoing'' edge that is greater than $uv$ or $\perp$ if such an edge does not exist, and
(3) sending the result back to $b$ via an \textsc{Aggregation}.

\paragraph{Augmentation and Change Base}
After obtaining the augmenting paths in $G/\Omega$, the corresponding augmenting path can be computed locally at each blossom's base.
After the augmentation, the old base $b$ informs the new base $b'$, and using $O(\Bmax^2)$ rounds to send all information that $b$ currently possesses.

\section{The Analysis of the Modified Scaling Algorithm}\label{apx:DPanalysis}
%\paragraph{The Invariants} 

The algorithm consists of $1 + \log_{2} W$ scales, where each scale consists of $O(1/\epsilon)$ iterations. We maintain the variables $M, y, z, \Omega,$ and $\Delta w$ so that they satisfy the relaxed complementary slackness conditions modified from \cite{DuanP14, CS22} at the end of each iteration of each scale $i$:

\begin{lemma}\label{lem:dif}
    Let $G_1 = (V,E,\hat{w})$ and $G_2 = (V,E,w)$. Let $M_1^{*}$ and $M_2^{*}$ be optimal matchings in $G_1$ and in $G_2$. Suppose that $\sum_{e \in E}|\hat{w}(e) - w(e)| \leq \gamma \hat{w}(M_1^{*})$ for some $\gamma$. We have $\hat{w}(M_2^{*}) \geq (1-2\gamma)\cdot \hat{w}(M_1^{*})$.
    \end{lemma}
\begin{proof}
We have
\begin{align*}
    \hat{w}(M_2^{*}) &= \sum_{e \in M_2^{*}} \hat{w}(e) \\
    &= \sum_{e\in M_2^{*}} (w(e) - (w(e) - \hat{w}(e) )) \\
    &\geq \sum_{e\in M_2^{*}} (w(e) - |w(e) - \hat{w}(e) | )\\
    &= w(M_2^{*}) - \gamma \hat{w}(M^{*}_1) \\
    &\geq w(M^{*}_1) - \gamma \hat{w}(M^{*}_1)  & \mbox{$M^{*}_2$ optimal w.r.t.~$w$}\\ 
    &\geq \hat{w}(M^{*}_1) - 2 \gamma \hat{w}(M^{*}_1) = (1-2\gamma)\cdot \hat{w}(M^{*}_1)  &&&& \qedhere
\end{align*}
\end{proof}

\begin{lemma}\label{lem:edge_weight_lower_bound} If $e$ is of type $j$, it must be the case that $w(e) \geq W/ 2^{j+1} + \delta_j$. At any point, if $e$ is a matched edge or a blossom edge and satisfies the near tightness condition then $yz(e) \leq (1+4\epsilon')w(e)$. If $e$ satisfies the near domination condition at the end of scale $L$, then $yz(e) \geq  w(e) - \epsilon'$.
\end{lemma}
\begin{proof}
 If $e$ is of type $j$, it must have become eligible at scale $j$ while unmatched. An unmatched edge can only become eligible in scale $j$ if $w_j(e) - \delta_j \geq yz(e)$. The minimum $y$-value over the vertices at scale $j$ is at least $W/2^{j+2}$. Therefore, $w(e)\geq w_{j}(e) \geq yz(e) + \delta_j \geq 2\cdot (W/2^{j+2}) + \delta_j = W/2^{j+1} + \delta_j$. 
 
 If $e$ is of type $j$ and satisfies the near tightness condition, then we must have $yz(e) \leq w_j(e) + 2\delta_j \leq w(e) + 2\delta_j \leq w(e) + 2\epsilon' W / 2^{j} \leq w(e) + 4\epsilon' w(e) = (1+4\epsilon')w(e)$. 
 
 If $e$ satisfies the near domination condition at the end of scale $L$, then $yz(e) \geq w_L(e) - \delta_L = w(e) - \epsilon'$.
\end{proof}

\begin{observation}\label{obs:aug_gone}
    If we augment along an augmenting path $P$ in $G_{elig}/\Omega$, all the edges of $P$ will become ineligible.
    \end{observation}

\begin{proof}[Proof of \Cref{lem:approximateMWM}.]
Let $M_1^{*}$ and $M_2^{*}$ be optimal matching with respect to $\hat{w}$ and $w$. First, we have
\begin{allowdisplaybreaks}
\begin{align*}
    w(M) &= \sum_{e \in M} w(e) \\
    &\geq \sum_{e \in M}(1+4\epsilon)^{-1} \cdot yz(e)  \hspace{57mm} \mbox{near tightness \& \cref{lem:edge_weight_lower_bound}} \\
    &= (1+4\epsilon')^{-1} \cdot \left( \sum_{u \in V} y(u) + \sum_{B \in \Omega} \frac{|B| - 1}{2} \cdot z(B) - \sum_{u \in \hat{F}}y(u)\right) \\
    &\geq (1+4\epsilon')^{-1} \left(\sum_{u \in V(M_2^{*})} y(u) + \sum_{B \in \Omega} (|M_2^{*} \cap E(B)|) \cdot z(B) \right) - (1+4\epsilon')^{-1}(\epsilon' \cdot \hat{w}(M_1^{*})) \\
    &\geq (1+4\epsilon')^{-1} \cdot \left( \sum_{e\in M_2^{*}} yz(e) \right) - 2\epsilon' \cdot \hat{w}(M_1^{*}) \\
    &\geq (1+4\epsilon')^{-1} \cdot  \left(\sum_{e \in M_{2}^{*}} w(e)  -  \epsilon' \right) - 2\epsilon' \cdot \hat{w}(M_1^{*}) \\
        &= (1+4\epsilon')^{-1} \cdot  \left(\sum_{e \in M_{2}^{*}}( \hat{w}(e) -  \epsilon' + (w(e) - \hat{w}(e)))   \right) - 2\epsilon' \cdot \hat{w}(M_1^{*})\\
        &\geq (1+4\epsilon')^{-1} \cdot  \left(\sum_{e \in M_{2}^{*}} (1-\epsilon')\cdot \hat{w}(e) - \sum_{e \in M_{2}^{*}} |(w(e) - \hat{w}(e))| \right) - 2\epsilon' \cdot \hat{w}(M_1^{*}) \\
        &\hspace{130.5mm} \hat{w}(e) \geq 1\\
        &\geq (1+4\epsilon')^{-1} \cdot  \left(\sum_{e \in M_{2}^{*}} ((1-\epsilon')\cdot \hat{w}(e)) - \epsilon'\hat{w}(M^{*}_1) \right) - 2\epsilon' \cdot \hat{w}(M_1^{*})\\
     &\geq (1+4\epsilon')^{-1}\cdot(1-\epsilon') \cdot (1-2\epsilon')  \hat{w}(M^{*}_1) - 4\epsilon' \cdot \hat{w}(M_1^{*}) \hspace{25mm} \mbox{By \cref{lem:dif}} \\
     &= (1-\Theta(\epsilon')) \hat{w}(M^{*}_1) = (1-\epsilon) \hat{w}(M^{*}_1)\qedhere
\end{align*}
\end{allowdisplaybreaks}
\end{proof} 

\begin{lemma}\label{lem:matching_size_weight} Let $M^{*}$ be an optimal matching in $G$ w.r.t. $\hat{w}$. At any point of the algorithm, $\hat{w}(M^{*}) \geq |M| \cdot (W / 2^{i+2}) = |M| \cdot (\delta_i / 4\epsilon)$. \end{lemma}
\begin{proof}
We have:
\begin{align*}
\hat{w}(M^{*}) &\geq \sum_{e \in M} \hat{w}(e)	 \\
&\geq \left(\sum_{e \in M} w(e)\right) -  \epsilon' \hat{w}(M^{*}) && \mbox{By \cref{prop:RCS}(\ref{item:boundedw})} \\
&\geq |M| \cdot (W / 2^{i+1}) - \epsilon' \hat{w}(M^{*}) && \mbox{By \cref{lem:edge_weight_lower_bound}}
\end{align*}
Therefore, $\hat{w}(M^{*}) \geq |M| \cdot (W/2^{i+1}) / (1+\epsilon') \geq W/2^{i+2}$.
\end{proof}

\begin{lemma}\label{lem:boundedchange} \cref{prop:RCS}(\ref{RCS:5})(\ref{item:boundedw}) holds throughout the algorithm. % At scale $i$, $\sum_{e \in E}| \Delta w(e)|$ change by at most $(K \cdot \log n) \cdot \lambda \cdot \hat{w}(M^{*})$. Moreover, the sum of the $y$-values of the free vertices decrease by at least $???$.
\end{lemma}

\begin{proof}
Note that there are $L = \log_{2} W + 1$ scales and there are at most $(1/\epsilon + 2)$ iterations per scale.

For each iteration at scale $i$, $\sum_{e \in E} \Delta w(e)$ increases by at most:
\begin{align*}\delta_i \cdot |M'| \leq \delta_i \lambda \cdot |M| \leq \lambda \cdot 4\epsilon \hat{w}(M^{*}) && \mbox{by \cref{lem:matching_size_weight}}\end{align*}
%For each component $Z$ in color class $\mathcal{C}$, the matched edges $M_{Z}$ returned by \cref{thm:augment_and_shrink} has size at most $\lambda |M(N^{R}(Z)|$.  Moreover, for any two different components $Z_1$ and $Z_2$ of color class $\mathcal{C}$, $N^{R}(Z_1)$ and $N^{R}(Z_2)$ are disjoint. The total increase on $\sum_{e \in E} \Delta w(e)$ for color class $\mathcal{C}$ is therefore upper bounded by:

%\begin{align*}\delta_i \cdot \sum_{Z:color(Z)=\mathcal{C}} |M_Z| \leq \delta_i \cdot  \sum_{Z:color(Z)=\mathcal{C}} \lambda |M(N^{R}(Z)| \leq \delta_i \cdot \lambda \cdot |M| \leq \lambda \cdot 4\epsilon \hat{w}(M^{*}) && \mbox{by \cref{lem:matching_size_weight}}\end{align*}

Since there are $L = \log_{2} W + 1$ scales, at most $(1/\epsilon + 2) \leq 3 /\epsilon$ iterations per scale, at any point of the algorithm $\sum_{e \in E} \Delta w(e)$ is upper bounded by: $$(1+\log_{2}W)(3/\epsilon) \cdot \lambda \cdot 4\epsilon \hat{w}(M^{*}) \leq  (1+\log_{2}W)\cdot 12 \cdot \lambda \cdot \hat{w}(M^{*}) \leq \epsilon' \hat{w}(M^{*})$$

Similar we have $|F'| \leq \lambda \cdot |M|$ by \cref{thm:augment_and_shrink}. During the iteration, every free vertex in $G$ besides those in $F'$ have its $y$-value decrease by $\delta_i / 2$. %Moreover, since by \cref{thm:augment_and_shrink} no free vertices of $G$ are in $\hat{V}_{out}(F^{\Omega}(Z))$, the $y$-values all free vertices outside of $F^{\Omega}(Z)$ remain unchanged.

%During an iteration, every free vertex in $F(G)$ besides those are frozen will have their $y$-values decrease by $\delta_i/2$. 

Thus, the  deviations of the $y$-values of the free vertices from their baseline, $(\sum_{f \in F(G)} y(f)) - |F(G)| \tau$, increase by at most $(\delta_i/2) \cdot \lambda \cdot |M| \leq 4\epsilon \hat{w}(M^{*})$ in an iteration.

Since there are $L = \log_{2} W + 1$ scales, at most $(1/\epsilon + 2) \leq 3 /\epsilon$ iterations per scale, at any point of the algorithm $(\sum_{f \in F(G)} y(f) - |F(G)| \tau)$ is upper bounded by: $$(1+\log_{2}W)(3/\epsilon) \cdot \lambda \cdot 4\epsilon \hat{w}(M^{*}) \leq  (1+\log_{2}W)\cdot 12 \cdot \lambda \cdot \hat{w}(M^{*}) \leq \epsilon' \hat{w}(M^{*})$$ 
\end{proof}

\begin{lemma}\label{lem:scale_begin_end}Suppose that  \cref{prop:RCS}(\ref{RCS:1})--(\ref{RCS:4}) holds at the end of scale $i$, they must also hold at the beginning of scale $i+1$. \end{lemma}

\begin{proof}
It is easy to see that \cref{prop:RCS}(\ref{RCS:1})(\ref{RCS:2}) are automatically satisfied as $\delta_{i+1} = \delta_{i} / 2$.  

For \cref{prop:RCS}(\ref{RCS:3}) (near domination), let $yz(e)$ and $yz'(e)$ denote the $yz$-value of $e$ prior to and after the update respectively in the last line of Scale $i$. We have:
\begin{align*}
yz'(e) &= yz(e) + 2\delta_{i+1} && \mbox{the $y$-value of both endpoints increase by $\delta_{i+1}$}\\
&\geq w_i(e) - \delta_{i} + 2\delta_{i+1} && \mbox{near domination at the end of scale $i$} \\
&\geq w_{i+1}(e) - \delta_{i+1} - \delta_{i} + 2\delta_{i+1}  \\
&= w_{i+1}(e) - \delta_{i+1} 
\end{align*}

For \cref{prop:RCS}(\ref{RCS:4}) (near tightness), let $yz(e)$ and $yz'(e)$ denote the $yz$-value of $e$ prior to and after the update respectively in the last line of Scale $i$. Suppose that $e$ is of type $j$. We have:
\begin{align*}
yz'(e) &= yz(e) + 2\delta_{i+1} && \mbox{the $y$-value of both endpoints increase by $\delta_{i+1}$}\\
&\leq w_i(e) + 2(\delta_j - \delta_i) + 2\delta_{i+1} && \mbox{near tightness at the end of scale $i$} \\
&\leq w_{i+1}(e) + 2(\delta_j - \delta_i + \delta_{i+1})  \\
&\leq w_{i+1}(e) + 2(\delta_j - \delta_{i+1}) &&&& \qedhere
\end{align*}
\end{proof}

\begin{lemma}\label{lem:parity}All vertices in $\hat{R}(F(G_{elig}/\Omega))$ have the same parity as a multiple of $\delta_i / 2$. \end{lemma}
\begin{proof}	
Initially, all vertices have the same $y$-value. We will argue inductively that the $y$-values of all the vertices in the same blossom have the parity, as a multiple of $\delta_i/2$. Suppose that it is true before we set $\Omega \leftarrow \Omega \cup \Omega'$ during an iteration. Consider an alternating path $P = (B_0, B_1,\ldots, B_k)$ in $G_{elig}$, where each $B_j$ is either a vertex or a blossom in $\Omega$ and $B_0 \in F(G_{elig}/\Omega)$. Let $(u_0,v_1), (u_1, v_2), \ldots, (u_{k-1}, v_k)$ be the edges in $G$ that correspond to the edges in $P$. Since $yz(u_{k-1}v_k)- w_i$ must be an integer multiple of $\delta_i$ by \cref{def:eligible}, $u_{k-1}$ and $v_{k}$ must have the same parity, as a multiple of $\delta_i / 2$.  This implies all the vertices in $\hat{B_0} \cup \hat{B_1} \ldots \cup \hat{B_k}$ have the same parity. 

Consider a blossom $\{{B}_1, \ldots {B}_{l}\}$ that has been formed, there must be a vertex $B_{f} \in F(G_{elig} /\Omega)$ such for every $B_j$ there is a an alternating path from $B_{f}$ to $B_j$. Therefore, all the vertices in the new blossoms have the same parity. Hence, the inductive hypothesis holds after setting $\Omega \leftarrow \Omega \cup \Omega'$.  

Now we argue that all vertices in $\hat{R}(F(G_{elig}/\Omega))$ have the same parity as a multiple of $\delta_i / 2$. For every vertex $B \in R(F(G_{elig} /\Omega))$, there must exist some vertex $B_{f} \in F(G_{elig}/ \Omega)$ such that there is an alternating path from $B_f$ to $B$. By the reasoning of an alternating path above and the fact that all vertices in the same blossom have the same parity, the parity of the vertices in $\hat{B}_f$ and $\hat{B}$ must be the same.  Since all the $y$-values of all free vertices have the same parity, all vertices in $\hat{R}(F(G_{elig}/\Omega))$ have the same parity.
\end{proof}

\begin{lemma}Suppose that \cref{prop:RCS}(\ref{RCS:1})--(\ref{RCS:4}) hold in the beginning of an iteration, they must also hold at the end of the iteration. \end{lemma}

\begin{proof}

Suppose that \cref{prop:RCS}(\ref{RCS:1})--(\ref{RCS:4}) hold right at the beginning of the iteration. We first show that \cref{prop:RCS}(\ref{RCS:1}) is satisfied after the iteration. First, note that throughout the iteration $z(B)$ and $\Delta w_i(e)$ change only by multiples of $\delta_i$ and $y(u)$ changes by multiples of $\delta_i/2$. Now we show that if $y(u)$, $w_i(e)$, $z(B)$ remain non-negative after the iteration. First $y(v) \geq \tau$ for every vertex $v$. The algorithm stops when $\tau$ reaches 0, so $y(v)$ is non-negative for every vertex $v$. For each edge $e$, since $\Delta w(e)$ only increases during the iteration, $w_i(e)$ must be non-negative. Consider a blossom $B \in \Omega$ after the Blossom Shrinking step. If $B$ is a root blossom that has been added to $\Omega$ during the iteration, then $B$ must be outer. This implies that $z(B) > 0$ after the Dual Adjustment step. If $B$ was a root blossom before the iteration, then $z(B) > 0$ and so $z(B) \geq \delta_i$ at the beginning of the iteration. Observe that $z(B)$ decreases by at most $\delta_i$ during an iteration: If $B$ contains a free vertex $f \in F'$ then $z(B)$ decreases by $2\delta_i$ during the Augmentation and Blossom Shrinking step. However, in this case, since $f \in \hat{V}_{out}(F(G_{elig}/\Omega)$, $Z(B)$ will increase by $\delta_i$ during the dual adjustment step so the net gain is $-\delta_i$, which implies $z(B) \geq 0$ before the Blossom Dissolution step. If $B$ does not contain a free vertex, then $z(B)$ decreases at most $\delta_i$ during the Dual Adjustment. If $B$ is a non-root blossom, then $z(B)$ remains unchanged after the dual adjustment step. In all three cases, $z(B)$ will be non-negative. 

Now we claim \cref{prop:RCS}(\ref{RCS:2}) is satisfied after the iteration. First we argue that each $B \in \Omega$ is {\it full} (i.e .$|M \cap E_{B}| = \lfloor |B| / 2\rfloor$). This is because an augmentation does not affect the fullness of the blossoms. Also, all the newly created blossoms in $\Omega_{Z}$ are full. %Moreover, note that the dummy vertices never join any blossoms during the iteration because they have no adjacent unmatched edges. So the removal of these vertices will not affect the fullness of the blossoms.
    
Secondly, all root blossoms must have $z(B) > 0$ after the iteration; otherwise, they would have been dissolved during the Blossom Dissolution step. Also, since we only dissolve blossoms with zero $z$-values and we only increase $z(B)$ if $B \in \Omega$, it must be the case that $z(B) = 0$ for $B \notin \Omega$.  

Now we show that \cref{prop:RCS}(\ref{RCS:3},\ref{RCS:4}) are satisfied. During an augmentation, when we switch the status of an eligible matched edge to unmatched or an eligible unmatched edge to matched, \cref{prop:RCS}(\ref{RCS:3}) and \cref{prop:RCS}(\ref{RCS:4}) are still guaranteed to hold. This is because changing a matched edge to unmatched does not impose additional constraints. On the other hand, an eligible unmatched edge $e$ at scale $i$ must satisfy $yz(e) = w_i(e) - \delta_i$. This means that if we change the status of $e$ to become matched, it also satisfies \cref{prop:RCS}(\ref{RCS:4}) (near tightness), which is required for matched edges.

Consider the step when we add $\delta_i$ to $\Delta w(e)$ for $e \in M'$ to remove it from $G_{elig}$. If $e$ is an eligible matched edge of type $j$, it satisfies $yz(e) = w_i(e) + 2(\delta_j - \delta_i)$. After we add $\delta_i$ to $\Delta w(e)$, we must have $yz(e) = w_i(e) + 2(\delta_j - \delta_{i}) - \delta_{i} \in [w_{i}(e) -\delta_i, w_{i}(e) + 2(\delta_j - \delta_i)]$. This means $yz(e)$ still satisfies \cref{prop:RCS}(\ref{RCS:3},\ref{RCS:4}).

Now we show that the dual adjustment step also maintains \cref{prop:RCS}(\ref{RCS:3},\ref{RCS:4}). Consider an edge $e=uv$. If both $u$ and $v$ are not in $\hat{V}_{in}(F(G_{elig}/\Omega)) \cup \hat{V}_{out}(F(G_{elig}/\Omega))$ or both $u$ and $v$ are in the same root blossom in $\Omega$, then $yz(e)$ is unchanged, which implies near tightness and near domination remains to be satisfied. The remaining cases are as follows:

\begin{enumerate}
\item $e \notin M$ and at least one endpoint is in $\hat{V}_{in}(F(G_{elig}/\Omega)) \cup \hat{V}_{out}(F(G/\Omega))$. If $e$ is ineligible, then $yz(e) > w_i(e) - \delta_i$. Since both $yz(e)$ and $w_i(e)$ are multiples of $\delta_i$ by \cref{lem:parity}, it must be the case that $yz(e) \geq w_i(e)$ before the adjustment. Since $yz(e)$ can decrease at most $\delta_i$ during the adjustment, we have $yz(e) \geq w_i(e) - 1.5\delta_i$ after the adjustment. If $e$ is eligible, then at least one of $u, v$ is in $\hat{V}_{in}(F(G_{elig}/\Omega))$ due to \cref{thm:augment_and_shrink} (i.e.~no eligible outer-outer edge exists in $G_{elig}/\Omega$). Therefore, $yz(e)$ cannot be reduced, which preserves \cref{prop:RCS}(\ref{RCS:3}).

\item $e \in M$ and at least one endpoint is in $\hat{V}_{in}(F(G_{elig}/\Omega)) \cup \hat{V}_{out}(F(G_{elig}/\Omega))$. If $e$ is ineligible, we have $yz(e) < w_i(e) + 2(\delta_j - \delta_i)$. Since both $yz(e)$ and $w_i(e)$ are multiples of $\delta_i$ by \cref{lem:parity}, it must be the case that $w_i(e) - \delta_i \leq yz(e) \leq w_i(e) + 2(\delta_j-\delta_i) - \delta_i$ before the adjustment. Since every vertex in $\hat{V}_{out}(F(G_{elig}/\Omega))$ is either free or incident to an eligible matched edge, we have $u,v \notin \hat{V}_{out}(F(G_{elig}/\Omega))$. Therefore, $yz(e)$ increased by at most $\delta_i$ and it cannot decrease during the adjustment, which implies $yz(e) \leq w_i(e) + 2(\delta_i - \delta_j)$ after the adjustment, satisfying \cref{prop:RCS}(\ref{RCS:3},\ref{RCS:4}) (near domination and near tightness). If $e$ is eligible then it must be the case that one endpoint is in $\hat{V}_{in}(F(G_{elig}/\Omega))$ and the other is in $\hat{V}_{out}(F(G_{elig}/\Omega))$. $yz(e)$ value would remain unchanged and so \cref{prop:RCS}(\ref{RCS:3},\ref{RCS:4}) are satisfied.\qedhere
%\item $e \notin M$ and only $v \notin \hat{V}_{in} \cup \hat{V}_{out}$.  
\end{enumerate}

\end{proof}

\begin{theorem}The algorithm in \Cref{fig:edmondssearch} outputs a $(1-\epsilon)$-approximate \MWM{}.  \end{theorem}

\begin{proof}
We claim that the relaxed complementary slackness condition is satisfied at the end of iteration $L$, so that by \cref{lem:approximateMWM}, $M$ is a $(1-\epsilon)$-approximate \MWM{}. First note by \cref{lem:boundedchange}, \cref{prop:RCS}(\ref{RCS:5},\ref{item:boundedw}) hold throughout the algorithm so we only need to show that \cref{prop:RCS}(\ref{RCS:1}--\ref{RCS:4}) hold. In the beginning of scale 0, as we set $y(u) \leftarrow W/2 - \delta_0 / 2$ for $u \in V$, $z(B) \leftarrow 0$ for all $B\subseteq V$, $\Delta w(e) \leftarrow 0$ for $e \in E$, and $\Omega, M \leftarrow \emptyset$, \cref{prop:RCS}(\ref{RCS:1}--\ref{RCS:4}) must be satisfied. Now note that by Lemma \ref{lem:scale_begin_end} if \cref{prop:RCS}(\ref{RCS:1}--\ref{RCS:4}) is satisfied at the beginning of the first iteration of scale $i$, it will be satisfied at the end of the last iteration of scale $i$. Moreover, if $i < L$, by \cref{lem:boundedchange}, \cref{prop:RCS}(\ref{RCS:1}--\ref{RCS:4}) will be satisfied at the beginning of scale $i+1$. Therefore, \cref{prop:RCS}(\ref{RCS:1}--\ref{RCS:4}) will be satisfied at the end of scale $L$ inductively.
\end{proof}

\begin{theorem} The algorithm in \Cref{fig:edmondssearch} can be implemented in $\poly(1/\epsilon, \log n)$ rounds in the \congest model. \end{theorem}
\begin{proof}

For each iteration, we argue that the Augmentation and Blossom Shrinking step, the Dual Adjustment step, and the Blossom Dissolution and Clean Up step can be implemented in $\poly(1/\epsilon, \log n)$ rounds.  

In the Augmentation and Blossom Shrinking step, the invocation of \textsc{Approx\_Primal}$(G_{elig},M,\Omega,  \lambda)$ takes $\poly(1/\epsilon, \log n)$ rounds by \cref{thm:augment_and_shrink}, which returns $\Psi$, $\Omega'$, $M'$, and $F'$ that are stored locally (see \cref{appendix:data-structures}). Moreover, it also marks whether each vertex is in $\hat{V}_{in}(F_{alive})$ or $\hat{V}_{out}(F_{alive})$. 

Now we explain how each step can be done in $\poly(1/\epsilon, \log n)$ rounds:

\begin{itemize}
\item Set $M \leftarrow M \oplus \bigcup_{P \in \Psi} P$.  Now each edge in $G_{elig}/\Omega$ knows whether it belongs to an augmenting path in $\Psi$ so it can flip its matching status in $O(1)$ round. For each root blossom $B$ in $\Omega$, if $B$ is in $P$, the flipping of an augmenting path inside $B$ and the change of base can be done in $O(\Bmax^2)$ rounds as shown in \cref{appendix:data-structures}.

\item Set $\Omega \leftarrow  \Omega \cup \Omega'$. Since each root blossoms in $\Omega'$ must be in inside $S_{\alpha}$ for some free vertex $\alpha$. The creation and storage of all the blossoms inside a root blossom can be done in $O(C^2_{\max}) = \poly(\log n, 1/\epsilon)$ rounds by  \cref{lemma:forming-blossoms} and \cref{appendix:data-structures}.

\item  For each $e \in M'$, set $\Delta w(e) + \delta_i$. Each edge knows whether it is in $M'$ so the adjustment of $\Delta w(e)$ takes $O(1)$ rounds. 

\item  For each $f \in F'$, set $y(u) \leftarrow y(u) + \delta_i$ for  $u \in \hat{f}^{\Omega}$ and $z(B) \leftarrow z(B) - 2\delta_i $ for the root blossom $B\ni f$ if it exists. These can be done in $O(\Bmax)$ rounds. 

\item Dual Adjustment. As the vertices have been marked whether they are in $\hat{V}_{in}(F(G_{elig}/\Omega))$ or $\hat{V}_{out}(F(G_{elig}/\Omega))$. The $y$-values can be adjusted in $O(1)$ rounds. The $z$-values of the blossoms can also be adjusted in $O(1)$ as they are stored at the base of the blossoms. 

\item Blossom Dissolution. As mentioned in \cref{appendix:data-structures}, it takes at most $O(\Bmax^2)$ rounds to perform the blossom dissolution step.
\end{itemize}

Since there are $O(\log W) = O(\log n)$ scales and each scale contains $O(1/\epsilon)$ iterations, the total running time is $\poly(\log n, 1/\epsilon)$.  
\end{proof}

\section{The {\normalfont\VertexAlgPhase} Algorithm}
\label{appendix:alg-phase}

The goal of this section is to prove \cref{lemma:vertex-weighted-FMU}.

\subsection{A Brief Overview to {\normalfont\textsc{Alg-Phase}} in~\cite{FMU22}}\label{sec:fmu-overview}

A single phase of the FMU algorithm, $\textsc{Alg-Phase}$, takes a graph $G$, a matching $M$, and a parameter $\epsilon$.
The \textsc{Alg-Phase} algorithm searches for a set of augmenting paths of lengths at most $\poly(1/\epsilon)$.
This is done via
performing
a specialized
DFS from all free vertices $\alpha\in F(G)$ simultaneously.
A \emph{structure}
$S_\alpha$ defined over a free vertex $\alpha\in F$ refers to the set of discovered vertices throughout DFS from $\alpha$:

\begin{definition}[{\cite[Definition 2.6 without active path]{FMU22}}]\label{def:structure}
A \emph{structure} $S_\alpha=(A_\alpha, P_\alpha)$ of a free node $\alpha$ is a set of arcs $A_\alpha$ along with an \emph{active path} $P_\alpha$ with the following properties:
\begin{enumerate}[itemsep=0pt]
  \item {\textbf{Disjointness}}: $S_\alpha$ is vertex-disjoint from all other structures. The set of vertices that are endpoints to any arc in $A_\alpha$ is denoted by $V_\alpha$.
  \item {\textbf{Alternating-paths to matched arcs}}: for each matched arc $\vec{e}$ there is an alternating path starting from $\alpha$ using only the arcs in $S_\alpha$.
  \item {\textbf{Endpoints of unmatched arcs}}: each endpoint of an unmatched arc of $S_\alpha$ is either $\alpha$ or an endpoint of a matched arc of $S_\alpha$.
\end{enumerate}
\end{definition}

The \emph{active path} associated with a structure is an alternating path starting from $\alpha$ and ending at the current visiting vertex (the \emph{head} of $S_\alpha$) in the DFS.
Moreover, each matched edge $e$ got at most two labels $\ell(\vec{e})$ and $\ell(\cev{e})$, one for each direction during the DFS. An edge with a specified direction is called an \emph{arc}. The labels of the arcs will be updated several times non-increasingly during the DFS, and these labels represent the shortest possible alternating path found currently from a free vertex.

The algorithm carefully maintains a collection of structures $\{S_\alpha\}$ such that they never overlap --- whenever a head of $S_\alpha$ is about to visit a new vertex (with an associated matched arc $\vec{e}$) that belongs to another structure $S_\beta$, one of the three things could happen:
(1) a new augmenting path is found,
(2) $\ell(\vec{e})$ is already smaller than the number of the matched edges on the current active path, or (3) $\ell(\vec{e})$ is at least the number of the matched edges on the current active path.
In case (1) we have found a new augmenting path in $S_\alpha\cup S_\beta$, the algorithm invokes \textsc{Augment-and-Clean} which 
adds the path to $\mathcal{P}$ and
removes all vertices in $S_\alpha\cup S_\beta$ from the graph.
In case (2) the head of $S_\alpha$ does not actually visit this new vertex. The DFS of $S_\alpha$ continues assuming that vertex in $S_\beta$ has been visited.
Case (3) is perhaps the crux of the entire algorithm: a \textsc{Reduce-Label-and-Overtake} procedure is invoked and some vertices, matched arcs and unmatched edges leave $S_\beta$ and join $S_\alpha$, and the DFS search on $S_\alpha$ ``fast-forwards'' for consuming all newly joined vertices and arcs.

In addition, this specialized DFS restricts itself to the search space and search depth.
The DFS starting from any free vertex does not go too deep and the whole structure does not grow too big.
Specifically, there is a \emph{depth limit} $\Lmax := 1/\epsilon$,
a \emph{size limit} $\limit:=1/\epsilon^2$,
and a \emph{pass limit} $\tau_{\max} := 1/\epsilon^4$ imposed to the DFS search\footnote{In~\cite{FMU22} their analysis requires a larger $\limit$ and $\tau_{\max}$ values 
because their algorithm has a quadratic relation between the approximation ratio and the ``almostness'' of an inclusion-maximal set of disjoint augmenting paths.}.
Notice that it is possible for an \textsc{Reduce-Label-and-Overtake} procedure to overtake an already large enough structure so that at the end the algorithm produces a structure of size much larger than $\textsc{limit}$.
Fortunately, with the pass limit, one can prove that the size of any structure does not exceed $\Cmax := \tau_{\max}\cdot(\Lmax+1)\cdot \limit = O(1/\epsilon^{7})$.
Moreover, by choosing a large enough pass limit and size limit, the number of free vertices that are still active after the search can be upper bounded by $h(\epsilon)|M|$, where $h(\epsilon):=\frac{4+2/\epsilon}{\epsilon\cdot \tau_{\max}} + \frac{2}{\limit} = O(\epsilon^2)$.

Throughout the execution of an \AlgPhase, the algorithm obtains a collection $\mathcal{P}$ of augmenting paths,
a set of vertices $V'$ to be removed,
as well as a collection of structures $\mathcal{S}$ which captures the DFS searching progress.
The algorithm also collects the active paths within each structure and put all active vertices into the set $V_A$.
The most important guarantee from the FMU algorithm~\cite[Lemma 4.1, full version]{FMU22} is that, once removing all vertices in $V'\cup V_A$ from $G$, all short augmenting paths (with at most $\Lmax$ matched edges) no longer exist.

The \AlgPhase algorithm can be summarized as follows:

\begin{lemma}[{\cite[Algorithm 2]{FMU22}}]\label{lemma:fmu}
  Let $G$ be the input graph, $M$ be the current matching,
  and $\epsilon$ be a parameter.
  Let the DFS parameters be $\tau_{\max}:=1/\epsilon^4$, $\limit:=1/\epsilon^2$, and $\Lmax:=1/\epsilon$.
  Then, there exists an algorithm \textsc{FMU-Alg-Phase} in both \congest model and \crewpram model such that in $\poly(1/\epsilon, \log n)$ time, returns $(\mathcal{P}, V', V_A, \mathcal{S})$ that satisfies the following:

\begin{enumerate}[itemsep=0pt]

\item \cite[Lemma 5.3, full version]{FMU22} The number of vertices in each structure $S_\alpha\in\mathcal{S}$ is at most $C_{\max}$. 
\item \cite[Derived from Section 6.4 and Section 6.7.4, full version]{FMU22} $|V'| \leq C_{\max} \cdot (2 |\mathcal{P}| + \epsilon^{32} \tau_{\max} |M|)$.

\item \cite[Lemma 5.5, full version]{FMU22} The number of free active vertices  $|F\cap V_{A}| \leq h(\epsilon) \cdot |M|$. Moreover, any active path contains at most $\Lmax$ matched arcs.

\item \cite[Lemma 4.1, full version]{FMU22} No augmenting paths $P$ with at most $\ell_{\max}$ matched edges exist in $G\setminus (V' \cup V_{A})$.

\item \cite[Implied by Induction Hypothesis in the proof of Lemma 4.1, full version]{FMU22} For each matched arc $\vec{e}$ such that $d_{G-V'-V_A, M}(\vec{e}) \le \Lmax$, then there exists a structure $S_\alpha\in\mathcal{S}$ such that $\vec{e}\in S_\alpha$.
\end{enumerate}

\end{lemma}

\paragraph{Remark} In order to guarantee Item 4 of \cref{lemma:fmu}, we slightly modify Algorithm 7 in \cite[full version]{FMU22} as the following:
(1) the $\poly(1/\epsilon)$ round $O(1)$-approximate maximum matching over an auxiliary graph $H$~\cite[Section 6.7.1, full version]{FMU22} is replaced by the $\poly(1/\epsilon, \log n)$ round maximal matching algorithm (e.g., \cite{Fischer17}).
(2) 
once we obtain a small enough sized maximal matching $|M_H| < \epsilon^{32}|M|/2$ over the auxiliary graph $H$, the structures that participate in the matching $M_H$ are all removed.
Therefore, the set $V'$ now includes slightly more vertices per \textsc{Pass-Bundle}, leading to Item 2 of \cref{lemma:fmu}. 

\subsection{The {\normalfont\VertexAlgPhase} Algorithm}
\label{appendix:proof-vertex-weighted-FMU}
\label{section:vertex-weighted-FMU}

Let $G'=(V/\Omega, E'):=G/\Omega$ be the contracted graph and $M'=M/\Omega$ be a matching on $G'$.
We assume the graph is \emph{vertex weighted} with respect to the collection of blossoms $\Omega$.
Since each blossom in $\Omega$ has a size at most $\Bmax$, we may assume that the vertex weights are at most $\Bmax$ as well.
That is, there is a positive integer weight function for all vertices $\|\cdot\| : V'\to \{1, 2, \ldots, {\Bmax}\}$.
In this subsection, we first assume that $G'$ is also the underlying communication graph, and establish a 
vertex-weighted \textsc{Alg-Phase} algorithm that proves \cref{lemma:vertex-weighted-FMU}.

In fact,
our vertex-weighted \textsc{Alg-Phase} algorithm is implemented via a simple reduction to the regular (unweighted) \textsc{Alg-Phase}: replacing each matched edge $e$ on $G'$ with an alternating path of $\|e\|_M \le \Bmax$ matched edges.

\begin{algorithm}[h]
  \caption{$\VertexAlgPhase(G':=G/\Omega, M, \lambda)$}\label{alg:vertex-weighted-fmu}
  \begin{algorithmic}[1]\small
  \Require A vertex-weighted graph $G'$, a matching $M$, and $\lambda$.
  \Ensure Collection of augmenting paths $\mathcal{P}$, the collection of all structures $\mathcal{S}$, the vertices to be removed $V'$, and the vertices in active paths $V_A$.
  
  \State Copy $G''\gets G'$.
  \State Remove all matched edge $e$ in $G''$ whenever $\|e\|_M > \Lmax$.\label{line:remove-heavy-edge}
  \State Remove all free vertices $\alpha$ in $G''$ whenever $(\|\alpha\|-1)/2 > \Lmax$.\label{line:remove-heavy-vertex}
  \State For each free vertex $\alpha$ with $(\|\alpha\|-1)/2 > 0$, extend with an alternating path of $(\|\alpha\|-1)/2$ matched edges.
  All edges that were originally incident to $\alpha$ now incident to the non-free endpoint of this alternating path.
  \State Replace each of the remaining matched edge $e$ in $G''$ with an alternating path $P_e$ (which begins and ends with matched edges) of length $2\|e\|_M+1$ (i.e., having exactly $\|e\|_M$ matched edges).
  \State Let $M''$ be the set of all matched edges in $G''$.
  \State $(\mathcal{P}'', \mathcal{S}'', V'', V''_A)\gets \AlgPhase(G'', M'', \lambda)$.\label{line:call-fmu} \Comment{See \cref{lemma:fmu}.}
  \State Construct $\mathcal{P}$,
  the set of augmenting paths (on $G'$) obtained from $\mathcal{P}''$ by recovering the replaced alternating paths back to a single matched edge.
  \State Construct $V'$ to contain all vertices in $V''\cap (V/\Omega)$. Furthermore, for every matched edge $e$, if $P_e\cap V''\neq\emptyset$ then we add \emph{both} endpoints of $e$ to $V'$.
  \State Similarly, construct $V_A$ to contain all vertices in $V''_A\cap (V/\Omega)$, as well as both endpoints of $e$ whenever $P_e\cap V''_A\neq\emptyset$. 
  \State $\mathcal{S}\gets \emptyset$.
  \For{each structure $S''_\alpha\in \mathcal{S}''$}
  \State Construct $S_\alpha$, which includes the free vertex $\alpha$ and all matched arcs $\vec{e}$ if \emph{the entire path $P_e$} is in $S_\alpha$.\label{line:construct-structure-remove-partial-path}
  \State Add all unmatched edges whose both endpoints are either $\alpha$ or incident to a matched arc in $S_\alpha$.
  \State $\mathcal{S}\gets \mathcal{S}\cup S_\alpha$.
  \EndFor
  \State \Return $(\mathcal{P}, \mathcal{S}, V', V_A)$.
  \end{algorithmic}
  \end{algorithm}

\begin{lemma}[Item 1 in \cref{lemma:vertex-weighted-FMU}]\label{lemma:3.1-1}
  Each set $S_\alpha\in \mathcal{S}$ returned by \cref{alg:vertex-weighted-fmu} is a structure and has at most $\Cmax$ vertices.
\end{lemma}

\begin{proof}
By Item 1 of \cref{lemma:fmu}, $S''_\alpha$ has at most $\Cmax$ vertices.
According to \cref{line:construct-structure-remove-partial-path}, all partially included alternating paths will be removed from $S_\alpha$. Hence, the total weight of vertices in $S_\alpha$ is no greater than the number of vertices in $S''_\alpha$, i.e., $\|S_\alpha\|\le |S''_\alpha| \le \Cmax$.
Also, it is straightforward to check that $S_\alpha$ satisfies all criteria of \cref{def:structure}, the definition of a structure.
\end{proof}

\begin{lemma}[Item 2 in \cref{lemma:vertex-weighted-FMU}]\label{lemma:3.1-2}
  $V'$ returned from \cref{alg:vertex-weighted-fmu} satisfies $$\|V'\|\le \Cmax\cdot (\Lmax+1)\cdot (2|\mathcal{P}| + \lambda^{32}\tau_{\max}|M|).$$
\end{lemma}

\begin{proof}
  At most $\Lmax$ more weights will join $V'$ per vertex in $V''$.
  Thus, by Item 2 of \cref{lemma:fmu} and the fact that $|M''| \le |M|$ we have 
  \begin{align*}
    \|V'\| &\le (\Lmax+1)|V''| \\
    &\le \Cmax (\Lmax+1)(2|\mathcal{P}| + \lambda^{32}\tau_{\max}|M''|)\\
    &\le \Cmax (\Lmax+1)(2|\mathcal{P}| + \lambda^{32}\tau_{\max}|M|)
  \end{align*}
    as desired.
\end{proof}

\begin{lemma}[Item 3 in \cref{lemma:vertex-weighted-FMU}]\label{lemma:3.1-3}
  $V_A$ returned from \cref{alg:vertex-weighted-fmu} satisfies $$\|V_A\| \le h(\epsilon)\cdot (2\Lmax) \cdot |M|.$$
\end{lemma}

\begin{proof}
At most $\Lmax$ more weights will join $V_A$ per active free vertex, since any partially joined alternating path has at most $\Lmax$ matched edges.
By Item 3 of \cref{lemma:fmu}, the total weight per active free vertex is at most $2\Lmax$.
Moreover, notice that the number of matched edges $M''$ in $G''$ on \cref{line:call-fmu} is at most $|M|$. Thus we have \begin{align*} 
  \|V_A\| &\le h(\epsilon)\cdot (2\Lmax)\cdot |M''| \\
  &\le h(\epsilon)\cdot (2\Lmax)\cdot |M|.\qedhere
\end{align*}
\end{proof}

\begin{lemma}[Item 4 in \cref{lemma:vertex-weighted-FMU}]\label{lemma:3.1-4}
No augmenting path $P$ with $\|P\|_M \le \Lmax$ exists in $G'\setminus (V'\cup V_A)$.
\end{lemma}

\begin{proof}
  Any augmenting path $P$ that includes a removed matched edge in \cref{line:remove-heavy-edge}
  or a removed free vertex in \cref{line:remove-heavy-vertex}
  must have $\|P\|_M >  \Lmax$.
  Hence, if there is an augmenting path $P$ with $\|P\|_M  \le \Lmax$, all edges and vertices are not removed from $G''$.

  Now, assume for contradiction that an augmenting path $P$ with $\|P\|_M$ exists on $G'\setminus (V'\cup V'_A)$.
  Let $P''$ obtained from $P$ by extending all matched edges and free vertices of $P$ with the corresponding alternating paths.  
  By Item 4 in \cref{lemma:fmu} 
  we know that $P''$ must intersect with $V''\cup V''_A$.
  However, this implies that $P$ intersects with  $V'\cup V_A$, a contradiction.
\end{proof}

\begin{lemma}[Item 5 in \cref{lemma:vertex-weighted-FMU}]\label{lemma:3.1-5}
Let $\vec{e}$ be a matched arc such that $d_{G'-V'-V_A, M}(\vec{e}) \le \Lmax$, then $\vec{e}$ belongs to a structure $S_\alpha\in\mathcal{S}$.
\end{lemma}

\begin{proof}
Consider the replaced alternating path $P_e$ of $e$ on $G''$.
We notice that $d_{G'-V'-V_A, M}(\vec{e}) \le \Lmax$ implies that every matched arc $\vec{f}$ on $P_e$ of the same direction as $\vec{e}$ also satisfies 
$d_{G''-V''-V''_A, M''}(\vec{f}) \le d_{G'-V'-V_A, M}(\vec{e}) \le \Lmax$.
By Item 5 in \cref{lemma:fmu}, we know that each of these matched arcs belongs to some structure $S''_\alpha$.
Let us assume that the last matched arc $\vec{f}$ on $P_e$ belongs to a structure $S''_\alpha$.
By \cref{def:structure}, there is an alternating path from $\alpha$ to $\vec{f}$ using only arcs within $S''_\alpha$.
This implies that the entire path $P_e$ belongs to $S''_\alpha$ as well.
Now, by \cref{line:construct-structure-remove-partial-path} of \cref{alg:vertex-weighted-fmu}, the arc $\vec{e}$ belongs to $S_\alpha$.
\end{proof}

The next lemma shows that it is possible to simulate \cref{alg:vertex-weighted-fmu} when the underlying communication graph is $G$.

\begin{lemma}\label{lemma:simulate-fmu}
The \congest simulation of $\VertexAlgPhase$ on $G'=G/\Omega$ over the network $G$ can be done with a $O(\Bmax^2)$ slow-down.
\end{lemma}

\begin{proof}
  First, we notice that suppose that $G'$ itself is the underlying communication network,
  then the round complexity $\VertexAlgPhase$ is $\poly(1/\epsilon, \log n)$, since each replaced alternating path can be simulated on the endpoints with only a $O(1)$ overhead. (The status on this alternating path can be always represented using only $O(\log n + \log\Lmax)$-bits.)

  %First, we notice that all operations in \VertexAlgPhase
  
  Finally, we show that the simulation of $\AlgPhase(G'', M'', \epsilon)$ over $G/\Omega$ can be done with a $O(\Bmax^2)$ slow-down.
  In Section 6 of \cite[full version]{FMU22}, the authors abstract out four types of operations (i) -- (iv) that constituted for simulating the $\AlgPhase$ algorithm on \congest.

  The operations (ii) traversing all edges, (iii) propagation across the same structure, (iv) aggregation for computing sizes of matching can be implemented in a straightforward way via the lazy implementation described in \cref{appendix:data-structures} with a $O(\Bmax^2)$ slow down.
  It now remains to show that (i) a $\poly(1/\epsilon, \log n)$ round $O(1)$-approximate maximum matching procedure exists on graphs $H$ which are built over a contracted graph $G/\Omega$.
  Since we are allowed to have $\poly(\log n)$ in our round complexity, applying a maximal matching procedure in~\cite{Fischer17} suffices for our purpose.
  One crucial step in~\cite{Fischer17} is to reduce the input graph into a graph of degree $\le 2$ via splitting vertices.
  The non-trivial part is, if we apply this vertex splitting procedure on $G/\Omega$, we have to make sure that the subsequent communications on the split vertices do not cause congestion in a blossom\footnote{This is not needed in \cite{FMU22}.}.
  Notice that we do not have sufficient time for gathering all incident edges' information into a single node in a blossom.
  Fortunately, 
  using the communication tree $T_B$ within a blossom mentioned in 
  \cref{sec:lazy-implementation}, we are able to greedily pair up all the incident edges such that the congestion between all paired incident edges is $1$. Thus, achieving only $O(\Bmax^2)$ slow-down on all operations.
\end{proof}

\begin{lemma}\label{lemma:simulate-fmu-pram}
The \crewpram simulation of $\VertexAlgPhase$ on $G'/\Omega$ over the input graph $G$ can be done with a $O(\Cmax^2\log n)$ slow-down.
\end{lemma}

\begin{proof}
Our \crewpram algorithm simulates the \congest implementation of \VertexAlgPhase described in~\Cref{lemma:simulate-fmu} in a straightforward way.
In most of the steps, each decision from each processor is done in constant time.
However, in each of the DFS step, the algorithm has to ensure progress.
That is, the algorithm is required to figure out the next available (and valid) unvisited vertex, which can be straightforwardly implemented in $O(\log \Delta)=O(\log n)$ time where $\Delta$ is the max degree of the graph, thereby an extra $\log n$ factor in the slow-down.
\end{proof}

\begin{proof}[Proof of \cref{lemma:vertex-weighted-FMU}]
It is straightforward to check that \cref{lemma:vertex-weighted-FMU}
is implied by \cref{lemma:3.1-1}, \cref{lemma:3.1-2}, \cref{lemma:3.1-3}, \cref{lemma:3.1-4}, \cref{lemma:3.1-5}, \cref{lemma:simulate-fmu}, and \cref{lemma:simulate-fmu-pram}.
\end{proof}

\section{$(1-\epsilon)$-\MWM{} in Semi-Streamming Model}\label{sec:semi-streaming-impl}

\newcommand{\MaxMinRatio}{max-min ratio\xspace}

This section aims to prove \Cref{thm:semi-streaming-result}.
The main ingredient we will be applying is the reduction from Gupta and Peng~\cite{GuptaP13}.
Let $G=(V, E, w)$ be the graph with posistive edge weights and 
let the \emph{\MaxMinRatio} $W$ be the ratio between the heaviest weight and the lightest weight.
Let $\epsilon < 1$ be a constant and let $\epsilon_0=\epsilon/5 < 1/5$.
Gupta and Peng~\cite{GuptaP13} reduces the $(1-\epsilon)$-\MWM{} problem with an arbitrary \MaxMinRatio $W$ into $O(\log_{1/\epsilon}W)$ instances of $(1-\epsilon_0)$-\MWM{} such that the \MaxMinRatio in each instance is bounded by $W'\le (1/\epsilon)^{O(1/\epsilon)}$.
We describe the reduction as follows.

For each edge $e\in E$, we define its \emph{bucket number} $b(e) := \lfloor \log_{1/\epsilon_0} w(e)\rfloor$. 
Let $k=\lceil 1/\epsilon_0\rceil$ and consider $k$ copies of the graph $G$: $G_0, G_1, \ldots, G_k$.
In the $i$-th copy, all edges whose bucket number modulo $k$ equals to $i$ are removed.
That is, the graph $G_i$ is composed of all the edges $E_i := \{e\in E\ |\ b(e)\not\equiv i\ (\!\bmod k)\}$.

Now, in each graph $G_i$, the consecutive buckets that were not removed get merged together and form a \emph{level}.
Specifically, we say an edge $e$ in $G_i$ has its level $\ell_i(e) := \lceil (b(e)-i)/k\rceil$.
These levels naturally partition each $G_i$ into $O(\frac{1}{k}\log_{(1/\epsilon_0)}W)$ subgraphs.
The edge weights at different levels differ by $\Omega(1/\epsilon_0)$ multiplicatively due to the removal of the buckets.
Let $H_{i, j}$ be the subgraph of $G_i$ that contains all level $j$ edges, and let $M_{i, j}$ be \emph{any} $(1-\epsilon_0)$-\MWM{} of $H_{i, j}$ computed by any algorithm.
For each graph $G_i$, by applying the classical greedy algorithm on $\cup_j M_{i, j}$, one obtains a matching $M_i$ that is an $(1-4\epsilon_0)$-\MWM{} on $G_i$.

To prove the approximation ratio, we use the fact that each $M_{i, j}$ is a matching and the edge weights between edges from $M_{i, j}$ and $M_{i, j-1}$ are differed by at least $(1/\epsilon_0)$.
For each edge $e$ presented in the resulting matching $M_i$, at most $2$ edges from each lower level of $\cup_j M_{i, j}$ may be blocked by $e$. 
Let $M^\star_i$ be a \MWM{} on $G_i$, and let $M^\star_{i, j}$ be a \MWM{} on $H_{i, j}$.
Hence, we have: 
\begin{align*}
w(M^\star_i) &\le \sum_j w(M^\star_{i, j}) \\
&\le \sum_j \frac{1}{1-\epsilon_0} w(M_{i, j}) \\
&\le \frac{1}{1-\epsilon_0} \sum_{e\in M_i} w(e)\left(1 + \sum_{j < \ell(e)} 2\epsilon_0^{\ell(e)-j}\right) \\
&=\frac{1+2\epsilon_0}{(1-\epsilon_0)^2} w(M_i) \le \frac{1}{1-4\epsilon_0}w(M_i). \tag{$\epsilon_0 < 1/5$}
\end{align*}

Finally, Gupta and Peng~\cite{GuptaP13} proved that if we choose the maximum weighted matching among $M = \arg\!\max_{M_i}\{w(M_1), w(M_2), \ldots, w(M_k)\}$, we obtain an $(1-5\epsilon_0)$-\MWM{} of $G$.
This is because the missing edges at each subgraph $G_i$ from $G$ forms a partition.
In particular, if we consider $M^\star$, a \MWM{} of $G$, each edge in $M^\star$ occurs in \emph{exactly} $k-1$ subgraphs.
Hence,
\begin{align*}
  (k-1) w(M^\star) &\le \sum_i w(M\cap G_i) \\
  &\le \sum_i w(M^\star_i) 
  \le \frac{1}{1-4\epsilon_0} \sum_i w(M_i) \\
  &\le \frac{k}{1-4\epsilon_0} w(M). \tag{$w(M_i)\le w(M)$}
\end{align*}

Therefore, we have $w(M^\star) \le [(1-1/k)(1-4\epsilon_0)]^{-1}w(M)$. Since $k=\lceil 1/\epsilon_0\rceil$, we have $1-1/k \ge 1-\epsilon_0$ and thereby 
$w(M^\star) \le (1-5\epsilon_0)^{-1} w(M)$
as desired.
We summarize the reduction described above in the following lemma, in terms of semi-streaming model:

\begin{lemma}\label{lemma:gupta-peng-reduction}
Let $\epsilon_0 < 1/5$ be a constant and let $\epsilon=5\epsilon_0$.
Suppose there exists a semi-streaming algorithm $\mathcal{A}$ such that, on a graph $G=(V, E, w)$ with $n$ vertices and \MaxMinRatio at most $(1/\epsilon_0)^{\lceil 1/\epsilon_0\rceil}$,
algorithm $\mathcal{A}$
outputs an $(1-\epsilon_0)$-\MWM{} in $\poly(1/\epsilon_0)$-passes and $n\cdot\poly(1/\epsilon_0)$ space.
Then, there exists a semi-streaming algorithm $\mathcal{A}'$ such that, on a graph $G=(V, E, w)$ with $n$ vertices and \MaxMinRatio $W$, algorithm $\mathcal{A'}$ outputs a $(1-\epsilon)$-\MWM{} in $\poly(1/\epsilon)$-passes and $n\cdot\log_{1/\epsilon}W\cdot \poly(1/\epsilon)$ space.
\end{lemma}

\begin{proof}
We note that the computation of the bucket number $b(e)$ and its levels $\ell_i(e)$ depends only on $\epsilon_0$. Using the reduction from Gupta and Peng~\cite{GuptaP13}, the algorithm $\mathcal{A}'$ simply simulates $\mathcal{A}$ on each of the $O(\log_{1/\epsilon_0}W)$ instances $\{H_{i, j}\}$. Since the matchings $M_{i, j}$ can be stored entirely using $n\cdot O(\log_{1/\epsilon} W)$ space, the output can be produced without any additional pass.
\end{proof}

Now, we turn our attention to solving a reduced instance:

\begin{lemma}\label{lem:semi-streaming-simulate-fig2}
Let $G=(V, E, w)$ be the input graph with $n$ vertices and \MaxMinRatio $W$.
Let $\epsilon <1/5$ be a constant.
There exists a semi-streaming algorithm simluating the modified scaling framework (\Cref{fig:edmondssearch}) in $\poly(1/\epsilon, \log W)$-passes and $n\cdot \poly(1/\epsilon, \log W)$ space.
\end{lemma}

\begin{proof}
First of all, it suffices to notice that $\lambda=O((1/\epsilon)\log W)=\poly(1/\epsilon, \log W)$.
Hence, the \textsc{Approx\_Primal} procedure (and the associated \VertexAlgPhase{}) can be simulated in $\poly(1/\lambda)$-passes and $n\cdot \poly(1/\lambda)$ space.

To analyze the number of passes to the entire scaling framework, we note that there are  $O(\log W)$ scales and within each scale there are $\poly(1/\epsilon)$ iterations. Therefore, the total number of passes needed is $\poly(1/\epsilon, \log W)$.

To analyze the space usage, it suffices to notice that 
each active blossom has size $\Bmax=\poly(1/\lambda)$ and each stored augmenting path has length $\poly(1/\lambda)$.
In addition, it is not hard to see that the set of active blossoms $\Omega$ is a laminar set.
This implies that there are at most $O(n)$ active blossoms.
Hence, the space usage for storing augmenting paths $\Psi$, $y$ values for vertices, and non-zero $z$ values for blossoms is at most $n\cdot \poly(1/\epsilon, \log W)$.
To bound the number of non-zero entries of edge modifiers $\Delta(e)$, it suffices to notice that any edge $e$ with $\Delta(e)\neq 0$ must be once a matched edge. In each iteration there can be at most $O(n)$ newly matched edges.
Therefore, the number of edges with non-zero edge modifiers can be bounded by $O(n)$ times the number of passes, which is $n\cdot \poly(1/\epsilon, \log W)$ as desired.
\end{proof}

\begin{proof}[Proof of \Cref{thm:semi-streaming-result}]
\Cref{thm:semi-streaming-result} follows after plugging \cref{lem:semi-streaming-simulate-fig2} into \cref{lemma:gupta-peng-reduction}, since now $W\le (1/\epsilon)^{O(1/\epsilon)}$ so $\log W = \poly(1/\epsilon)$.
\end{proof}

\end{document}